\documentclass[11pt,letterpaper]{article}

\usepackage[margin=1in]{geometry}
\usepackage{amsmath,amsthm,amsfonts,amssymb,amscd,bm,mathrsfs}
\usepackage{fancyhdr}
\usepackage{xcolor}
\usepackage{graphicx}
\usepackage[colorlinks=true, urlcolor=blue, linkcolor=blue, citecolor=magenta]{hyperref}
\usepackage{multirow}
\usepackage{mathpazo}
\usepackage{microtype}
\usepackage{enumitem}
\usepackage[T1]{fontenc}
\usepackage{amsthm}
\usepackage{aliascnt}
\usepackage{tikz}
\usetikzlibrary{arrows.meta}
\usepackage{placeins}

\newlength{\oldparskip}
\newlength{\parskipdelta}

\setlist[itemize,1]{
    before=\vspace{-\parskipdelta}
}

\newlength{\targetparskip}
\newcommand{\useoldparskipforheading}{%
    \setlength{\parskip}{\oldparskip}%
    \AddToHookNext{para/begin}{%
        \setlength{\parskip}{\targetparskip}%
    }%
}

\AddToHook{cmd/paragraph/before}{\useoldparskipforheading}
\AddToHook{cmd/subparagraph/before}{\useoldparskipforheading}

\DeclareTextFontCommand{\textsc}{\upshape\scshape}

\allowdisplaybreaks

\usepackage[capitalise,nameinlink]{cleveref}
\Crefname{lemma}{Lemma}{Lemmas}
\Crefname{fact}{Fact}{Facts}
\Crefname{theorem}{Theorem}{Theorems}
\Crefname{corollary}{Corollary}{Corollaries}
\Crefname{claim}{Claim}{Claims}
\Crefname{example}{Example}{Examples}
\Crefname{problem}{Problem}{Problems}
\Crefname{definition}{Definition}{Definitions}
\Crefname{notation}{Notation}{Notations}
\Crefname{assumption}{Assumption}{Assumptions}
\Crefname{hypothesis}{Hypothesis}{Hypotheses}
\Crefname{subsection}{Subsection}{Subsections}
\Crefname{section}{Section}{Sections}
\Crefformat{equation}{(#2#1#3)}

\newtheorem{theorem}{Theorem}[section]
\newtheorem*{theorem*}{Theorem}

\newcommand{\newaliastheorem}[2]{%
  \newaliascnt{#1}{theorem}%
  \newtheorem{#1}[#1]{#2}%
  \aliascntresetthe{#1}%
  \crefname{#1}{#2}{#2s}%
  \Crefname{#1}{#2}{#2s}%
}

\newaliastheorem{proposition}{Proposition}
\newtheorem*{proposition*}{Proposition}

\newaliastheorem{lemma}{Lemma}
\newtheorem*{lemma*}{Lemma}

\newaliastheorem{corollary}{Corollary}
\newtheorem*{corollary*}{Corollary}

\newaliastheorem{fact}{Fact}
\newtheorem*{fact*}{Fact}

\newaliastheorem{exercise}{Exercise}
\newtheorem*{exercise*}{Exercise}

\newaliastheorem{hypothesis}{Hypothesis}
\newtheorem*{hypothesis*}{Hypothesis}

\newaliastheorem{conjecture}{Conjecture}
\newtheorem*{conjecture*}{Conjecture}

\newaliastheorem{subclaim}{Claim}
\newaliastheorem{claim}{Claim}
\newtheorem*{claim*}{Claim}

\theoremstyle{definition}
\newaliastheorem{definition}{Definition}
\newaliastheorem{notation}{Notation}
\newaliastheorem{construction}{Construction}
\newaliastheorem{example}{Example}
\newaliastheorem{question}{Question}
\newaliastheorem{openquestion}{Open Question}
\newaliastheorem{problem}{Problem}
\newaliastheorem{protocol}{Protocol}
\newaliastheorem{assumption}{Assumption}

\newaliastheorem{remark}{Remark}
\newtheorem*{remark*}{Remark}

\newaliastheorem{observation}{Observation}
\newtheorem*{observation*}{Observation}

\newaliastheorem{openproblem}{Open Problem}

\usepackage{microtype}

\DeclareMathOperator*{\E}{\mathbb E}

\renewcommand{\Pr}{\operatorname*{\mathbf{Pr}}}

\newcommand{\AND}[1]{\bigwedge_{\substack{#1}}}
\newcommand{\OR}[1]{\bigvee_{\substack{#1}}}

\newcommand{\e}{\mathrm{e}}

\newcommand{\poly}{\mathsf{poly}}

\newcommand{\F}{\mathbb{F}}
\newcommand{\N}{\mathbb{N}}

\newcommand{\Acal}{\mathcal{A}}

\newcommand{\Ccal}{\mathcal{C}}

\newcommand{\Hcal}{\mathcal{H}}

\newcommand{\Lcal}{\mathcal{L}}

\newcommand{\Ocal}{\mathcal{O}}
\newcommand{\Pcal}{\mathcal{P}}

\newcommand{\Scal}{\mathcal{S}}
\newcommand{\Tcal}{\mathcal{T}}

\newcommand{\Xcal}{\mathcal{X}}

\newcommand{\union}{\mathbin{\dot\cup}}

\newcommand{\opt}{\mathsf{opt}}

\newcommand{\FPT}{\ensuremath{\mathsf{FPT}}}

\newcommand{\Wone}{\ensuremath{\mathsf{W}[1]}\xspace}
\newcommand{\Wtwo}{\ensuremath{\mathsf{W}[2]}\xspace}
\newcommand{\ETH}{\ensuremath{\mathsf{ETH}}\xspace}

\usepackage{xspace}
\newcommand{\setcover}{\textsc{SetCover}\xspace}

\newcommand{\labelcover}{\textsc{LabelCover}\xspace}
\newcommand{\mmsalevel}[1]{\texorpdfstring{\textsc{MMSA}\ensuremath{{}_{#1}}}{MMSA#1}\xspace}

\newcommand{\vectorsum}{\textsc{VectorSum}\xspace}
\newcommand{\clique}{\textsc{Clique}\xspace}

\newcommand{\ov}{\textsc{OV}\xspace}

\title{Near-Logarithmic Inapproximability of Parameterized Set Cover}
\author{
  Bingkai Lin%
  \thanks{
  Nanjing University, Nanjing, China.
  Email: \texttt{lin@nju.edu.cn}}
  \and
  Xin Zheng%
  \thanks{
  Nanjing University, Nanjing, China.
  Email: \texttt{xinzheng@smail.nju.edu.cn}}
}
\date{}

\hypersetup{
    pdftitle={Near-Logarithmic Inapproximability of Parameterized Set Cover},
    pdfauthor={Bingkai Lin and Xin Zheng}
}

\begin{document}
    \hypersetup{pageanchor=false}
    \maketitle
    \begin{abstract}
        In the $\gamma$-gap $k$-\setcover problem, we are given a family $\Scal$ over a universe $U$, and the goal is to distinguish between the case where $\Scal$ contains $k$ sets covering $U$, and the case where at least $\gamma k$ sets in $\Scal$ are needed to cover $U$.

We show that, for every $k^{\Omega(k^2)}\le\rho\le n^{O(1)}$, it is $\Wone$-hard to approximate $k$-\setcover within ratio $\Omega\big(\frac{\log\rho}{k^2\log\log\rho}\big)$ on instances with $n$ sets and universe size at most $\rho\log n$.
Taking $\rho=n^{\Theta(1)}$ rules out $o\big(\frac{\log N}{\log\log N}\big)$-approximation FPT algorithms where $N$ denotes the input length.
We also show that under the randomized $k$-OV Hypothesis, there is no bounded-error randomized $f(k)n^{o(k)}$-time $o\big(\frac{\log N}{\log\log N}\big)$-approximation algorithm for $k$-\setcover.
These results improve on the previous $(\log n)^{1/O(k)}$-inapproximability.
For fixed $k$, this matches the polynomial-time $(\ln\rho+O(1))$-approximation up to a $\log\log\rho$ factor.

We further establish a connection between small universes and large approximation gaps.
We show that for every $1\le\gamma\le\min\big\{\frac{\log u}{\log\log N},\frac{\log u}{k^2\log\log u}\big\}$, exact $k$-\setcover with universe size $u^{\Theta(1/\gamma)}$ is equivalent to $\gamma$-gap $k$-\setcover with universe size $u$ under randomized Turing reductions.
In particular, superconstant gaps correspond to subpolynomial universe sizes.
This suggests a barrier to proving superconstant-gap $\Wtwo$-hardness: such a result would also make exact $k$-\setcover with universe size $n^{o(1)}$ $\Wtwo$-hard under randomized FPT Turing reductions.

Finally, we analyze the VC dimensions of the set systems produced by our reduction.
Write $d$ and $d^*$ for the primal and dual VC dimensions respectively.
We show that for $\gamma=o(\frac{\log N}{\log\log N})$, $\gamma$-gap $k$-\setcover is $\Wone$-hard even when $d=O(k^{2}\gamma\log(k\gamma))$ and $d^*=O(k\gamma)$.
The bound for $d^*$ matches the $O(d^*\log k)$-approximation algorithm up to a $k\log k$ factor.

    \end{abstract}
    \thispagestyle{empty}
    \clearpage
    \pagenumbering{roman}
    \hypersetup{pageanchor=true}
    \setcounter{tocdepth}{2}
    \begingroup
    \small
    \setlength{\parskip}{0pt}
    \tableofcontents
    \endgroup
    \clearpage
    \pagenumbering{arabic}

    \section{Introduction}
\label{sec:introduction}

\setcover is a central problem in combinatorial optimization. Given a universe $U$ and a family $\Scal$ of subsets of $U$, the goal is to find a smallest subfamily of $\Scal$ whose union is $U$. We write $n:=|\Scal|$ and $u:=|U|$, and assume throughout that $u\le n^{O(1)}$.

The polynomial-time approximability of \setcover is well understood. The classical greedy algorithm achieves an approximation ratio of $\ln u-\ln\ln u+O(1)$~\cite{Joh74,Lov75,Chvatal79,Sla97}. On the hardness side, Lund and Yannakakis~\cite{LY94} proved inapproximability within a factor of $0.36\ln u$ under the assumption $\mathsf{NP}\not\subseteq\mathsf{DTIME}[n^{\poly\log n}]$, and Feige~\cite{Feige98} improved the factor to $(1-o(1))\ln u$ under the weaker assumption $\mathsf{NP}\not\subseteq\mathsf{DTIME}[n^{O(\log\log n)}]$. The PCP results of Raz and Safra~\cite{RS97} and Arora and Sudan~\cite{AS03} yielded NP-hardness for logarithmic approximation factors, and Alon, Moshkovitz, and Safra~\cite{AMS06} improved the constant. Finally, the reduction of Moshkovitz~\cite{DBLP:journals/toc/Moshkovitz15}, combined with the parallel repetition theorem of Dinur and Steurer~\cite{DS14}, established NP-hardness of approximation within $(1-\varepsilon)\ln u$ for every fixed $\varepsilon>0$, matching the greedy guarantee in the leading constant.

We study \setcover through the lens of parameterized complexity, with the target cover size $k$ as the parameter. An $\FPT$ algorithm runs in time $f(k)\cdot n^{O(1)}$ for some computable function $f$. For a factor $\gamma\ge1$, the $\gamma$-gap $k$-\setcover problem asks us to distinguish $\opt\le k$ from $\opt>\gamma k$, where $\opt$ denotes the minimum number of sets needed to cover $U$.

\paragraph{Known lower bounds.}
Chen and Lin~\cite{ChenLin16} showed that, assuming $\Wone\ne\FPT$, no $\FPT$ algorithm approximates $k$-\setcover within any constant factor. Karthik, Laekhanukit, and Manurangsi~\cite{DBLP:journals/jacm/SLM19} ruled out $\FPT$ algorithms with approximation ratio $(\log n)^{1/k^{O(1)}}$ under the same assumption, and Lin~\cite{DBLP:conf/icalp/Lin19} improved the ratio to $(\log n)^{1/O(k)}$. Karthik and Livni-Navon~\cite{DBLP:conf/sosa/SN21} subsequently gave a general construction of the underlying threshold graphs from error-correcting codes.

These works also give fine-grained lower bounds. Karthik et al.~\cite{DBLP:journals/jacm/SLM19} ruled out $(\log n)^{1/k^{O(1)}}$-approximation algorithms running in $f(k)n^{o(k)}$ time under \ETH. They established the same ratio under the $k$-SUM Hypothesis, ruling out $f(k)n^{\lceil k/2\rceil-\varepsilon}$-time algorithms for every constant $\varepsilon>0$, and Lin~\cite{DBLP:conf/icalp/Lin19} improved the ratio to $(\log n)^{1/O(k)}$ under both assumptions. The recently announced refutation of \ETH~\cite{OpenAI26ETH} removes the \ETH-based justification for these lower bounds, but would not by itself invalidate their counterparts based on $k$-SUM hardness for sufficiently large $k$.

All of these ratios, such as $(\log n)^{1/\poly(k)}$, remain far below the $O(\log n)$ guarantee of the greedy algorithm. Lin, Ren, Sun, and Wang~\cite{DBLP:conf/soda/LinRSW23} made progress in two directions. First, they proved that approximation within any fixed constant is $\Wtwo$-hard, establishing constant-factor $\FPT$ inapproximability under the weaker assumption $\Wtwo\ne\FPT$. Second, they ruled out polynomial-time $o(\log n/\log\log n)$-approximation under $\Wone\ne\FPT$, even when $k=O((\log n)^3)$. This near-logarithmic lower bound does not, however, give $\FPT$ inapproximability parameterized by $k$, since $k$ grows with the input size. This raises the central question of this work, which is also an open problem highlighted by~\cite{FKLM20}: can near-logarithmic inapproximability be established against arbitrary $\FPT$ algorithms, and can such hardness hold even against $f(k)n^{o(k)}$-time algorithms under an alternative fine-grained hypothesis?

\paragraph{Approximation versus universe size.}
We also study the tradeoff between the approximation ratio and the universe size. The greedy guarantee depends on $\log u$, whereas the parameterized lower bounds above are expressed in terms of $\log n$. These do not match when $u=n^{o(1)}$: for example, when $u=(\log n)^{O(1)}$, the greedy algorithm already achieves an $O(\log\log n)$ approximation. This motivates studying inapproximability on instances with a small universe: how strong can the lower bounds remain as $u$ decreases, and how closely can they match the logarithmic dependence on $u$ in the greedy guarantee?

\paragraph{VC dimension.}
Finally, we study approximation hardness under bounds on the primal and dual VC dimensions.
The primal VC dimension measures the ability of sets to distinguish elements of the universe, while the dual dimension reverses the roles of sets and elements.
Low VC dimension often enables better approximation algorithms~\cite{BG95,ChekuriCH09}, but it does not make the exact problem easy: Bringmann, Kozma, Moran, and Narayanaswamy~\cite{BKMN16} proved that exact Hitting Set is $\Wone$-hard even when both dimensions are at most two, and by duality the same holds for exact $k$-\setcover. Kundu et al.~\cite{KSSU26} proved $\Wone$-hardness for every fixed factor $1<\gamma<2$ even when the VC dimension is at most $7$. We investigate how much approximation hardness persists while simultaneously controlling both dimensions.

\subsection{Our Results}

\paragraph{Near-Logarithmic Inapproximability of $k$-\setcover.}
Our main results establish near-logarithmic gap hardness for parameterized \setcover under $\Wone\ne\FPT$. The following theorem gives a lower bound in terms of the universe size.

\begin{theorem}
\label{thm:intro-small-universe-w1}
    There exists a constant $\alpha>0$ such that,
    for every \FPT-computable $\rho=\rho(k,n)$ satisfying $k^{\Omega(k^2)}\le\rho\le n^{O(1)}$,
    $$
        \frac{\alpha\log\rho}{k^2\log\log\rho}\text{-gap $k$-\setcover}
    $$
    on instances $(\Scal,U)$ with $|\Scal|=n$ and $|U|\le\rho\log n$ is $\Wone$-hard.
\end{theorem}

There is a polynomial-time $(\ln\rho+O(1))$-approximation algorithm for $k$-\setcover with $|U|\le\rho\log n$:
we first greedily cover the universe until at most $\log n$ elements remain.
If $\opt\le k$, after $k\ln\rho$ greedy choices, the number of uncovered elements is at most $(1-1/k)^{k\ln\rho}|U|\le\log n$, thus this step uses at most $k\ln\rho$ sets.
We then solve the remaining instance exactly by dynamic programming over subsets of the remaining universe, which takes $2^{O(\log n)}=n^{O(1)}$ time, and uses at most $k$ additional sets.
The total number of sets is at most $(\ln\rho+O(1))k$.
Therefore, for fixed $k$, the $\Omega_k\big(\frac{\log\rho}{\log\log\rho}\big)$-inapproximability in \cref{thm:intro-small-universe-w1} is tight up to a $\log\log\rho$ factor.

Taking $\rho=n^{\Theta(1)}$, so that the universe has polynomial size, gives the following corollary.

\begin{corollary}
\label{cor:intro-near-logarithmic-w1}
    There exists a constant $\alpha>0$ such that $\frac{\alpha\log n}{k^2\log\log n}$-gap $k$-\setcover is $\Wone$-hard.
\end{corollary}

We remark that the $\Omega_k\big(\frac{\log n}{\log\log n}\big)$-factor $\Wone$-hardness for $k$-\setcover was obtained independently and almost simultaneously by Guruswami and Ren~\cite{GR26} and Karthik~\cite{Karthik26}.
The reduction in~\cite{GR26} follows a framework similar to ours: it first compresses the 2-CSP alphabet using perfect hashing and then constructs a set system from unsatisfiable templates, analogously to our local acceptance family construction.
Our results additionally quantify the inapproximability in terms of the universe size and apply even when $|U|=\omega_k(\log n)$.

\paragraph{The Universe--Gap Tradeoff.}
A natural further goal is to understand whether superconstant approximation factors can be $\Wtwo$-hard, a question raised in~\cite{DBLP:conf/soda/LinRSW23}. Towards this goal, we study the relation between small-universe and large-gap $k$-\setcover instances. Lin~\cite{DBLP:conf/icalp/Lin19} used gap-producing self-reductions to turn exact hardness on small universes into approximation hardness. We improve this result with a randomized self-reduction from exact $k$-\setcover with universe size $u^{\Theta(1/\gamma)}$ to $\gamma$-gap $k$-\setcover with universe size $u$, and we complement it with a randomized $\FPT$ Turing reduction in the reverse direction. Together, these reductions give the following equivalence.

\begin{theorem}
\label{thm:intro-equivalence}
    Let $n,u,k\in\N$ and $\gamma$ satisfy $u\le n$ and
    $$
        1\le \gamma\le \min\bigg\{
            \frac{\log u}{\log\log n},
            \frac{\log u}{k^2\log\log u}
        \bigg\}.
    $$
    Exact $k$-\setcover with universe size $u^{\Theta(1/\gamma)}$ and $\gamma$-gap $k$-\setcover with universe size $u$ are equivalent under randomized $\FPT$ Turing reductions.
\end{theorem}

In particular, $\omega(1)$-gap $k$-\setcover is equivalent to $n^{o(1)}$-universe $k$-\setcover under randomized $\FPT$ Turing reductions. Hence, superconstant-gap $k$-\setcover is not likely to be $\Wtwo$-hard unless the exact $n^{o(1)}$-universe $k$-\setcover is $\Wtwo$-hard.

The forward reduction in the above theorem also yields fine-grained hardness of approximating $k$-\setcover.
Indeed, $k$-\setcover with universe size $O(\log n)$ is equivalent to $k$-\ov in dimension $O(\log n)$, and hence admits no $n^{k-\varepsilon}$-time algorithm under the randomized $k$-\ov Hypothesis.
Combining this hardness with our reduction gives the following theorem.

\begin{theorem}
\label{thm:intro-kov-hardness}
    Assume the randomized $k$-\ov Hypothesis.
    There are constants $\alpha,\varepsilon>0$ such that no bounded-error randomized algorithm solves
    $$
        \frac{\alpha\log N}{k^2\log\log N}
        \text{-gap $k$-\setcover}
    $$
    on instances $(\Scal,U)$ with $|\Scal|=N$ and $|U|\le N$ in $T(k)N^{\varepsilon k}$ time, for any computable function $T$.
\end{theorem}

\paragraph{Bounded VC Dimensions.}
Analyzing the primal and dual VC dimensions of the resulting $k$-\setcover instance given by \cref{thm:intro-small-universe-w1}, we have the following theorem.

\begin{theorem}
\label{thm:intro-vc-hardness}
    There exists a constant $\alpha>0$ such that, for every \FPT-computable $\gamma=\gamma(k,N)$ satisfying
    $$
        1\le\gamma\le\frac{\alpha\log N}{k^2\log\log N},
    $$
    $\gamma$-gap $k$-\setcover is $\Wone$-hard on instances $(\Scal,U)$ with $|\Scal|=N$, $|U|=N^{O(1)}$,
    and primal and dual VC dimensions satisfying $d=O(k^2\gamma\log(k\gamma))$ and $d^*=O(k\gamma)$.
\end{theorem}

Standard LP rounding via $\varepsilon$-nets gives a randomized polynomial-time approximation for $k$-\setcover within ratio $O(d^*\log k)$~\cite{ChekuriCH09}. Our lower bound $\gamma=\Omega(d^*/k)$ matches the algorithm up to a factor of $O(k\log k)$.
For comparison, \cite{KSSU26} proved $\Wone$-hardness for $\gamma$-gap $k$-\setcover with constant VC dimension when $\gamma<2$. Our results do not yield constant VC dimension for small $\gamma$, but establish hardness for larger gaps with bounds on the VC dimension.

\subsection{Proof Overview}

\subsubsection{From \labelcover to \ensuremath{k}-\setcover via Acceptance Families}
We start from the \labelcover formulation of~\cite{DBLP:journals/jacm/SLM19} and give a gap-preserving reduction to $k$-\setcover. This yields the near-logarithmic $\Wone$-hardness of $k$-\setcover.

A \labelcover instance is a bipartite graph with left vertex set $A=A_1\union\cdots\union A_k$ and right vertex set $B=B_1\union\cdots\union B_m$, where each $A_i$ is a copy of $[n]$ and each $B_r$ is a copy of $[t]$.
We call each $B_r$ a test.
A selection $X\subseteq A$ of left vertices passes test $r$ if some $b\in B_r$ has a neighbor in $X\cap A_i$ for every $i\in[k]$.
Let $\opt$ denote the minimum size of a selection $X\subseteq A$ passing every test. The $(k,h)$-gap \labelcover problem asks whether $\opt\le k$ or $\opt>h$.

The reduction in~\cite{DBLP:journals/jacm/SLM19} shows that $\frac{\log n}{k^{O(1)}\log\log n}$-gap \labelcover is $\Wone$-hard, and the resulting instances have a useful property:
\begin{itemize}
    \item Let $P_{r,a}:=N(a)\cap B_r$ denote the neighbors of $a\in A$ in test $r$. For a fixed test $r$, the number of distinct sets $P_{r,a}$ is at most $(\log n)^{k^{O(1)}}$.
\end{itemize}
We partition the vertices by putting $a,a'\in A_i$ in the same group if and only if $P_{r,a}=P_{r,a'}$, since whether the test passes depends only on which groups contain a selected vertex.
The grouping may be different in different tests, but we only consider one test now.
We view each group in $A_i$ as a new vertex $\sigma\in\Sigma_i$, and let $\Sigma:=\Sigma_1\cup\dots\cup\Sigma_k$ and $p:=|\Sigma|=(\log n)^{k^{O(1)}}$.

Fix a test $r$.
Let $\Acal\subseteq\binom{\Sigma}{k}$ consist of all $\{\sigma_1,\dots,\sigma_k\}$ passing test $r$, where $\sigma_i\in \Sigma_i$.
By definition, $X\subseteq \Sigma$ passes test $r$ if and only if it contains some $Z\in\Acal$.
To convert test $r$ into an equivalent $k$-\setcover instance for $|X|\le h$, we need a set family $(\{S_\sigma\}_{\sigma\in \Sigma},U)$ such that
$$
    \bigcup_{\sigma\in X}S_\sigma=U
    \quad\Longleftrightarrow\quad
    \exists Z\in\Acal:\ Z\subseteq X
    \qquad\text{for every }|X|\le h.
$$
We call this gadget an $h$-local acceptance family.
Each vertex $\sigma\in \Sigma$ is associated with the set $S_\sigma$, so that for $|X|\le h$, the selection $X$ passes test $r$ if and only if the sets $\{S_\sigma\}_{\sigma\in X}$ cover $U$.
Constructing such a family $(\{S^{(r)}_\sigma\}_{\sigma\in \Sigma},U^{(r)})$ for every test and taking the disjoint union gives a gap-preserving reduction to $k$-\setcover: the universe is $U^{(1)}\union\dots\union U^{(m)}$ and the set associated with $a\in A$ is $S^{(1)}_{\sigma_1}\union\dots\union S^{(m)}_{\sigma_m}$, where $\sigma_r$ denotes the group index of $a$ in the test $r$.

\paragraph{The construction.}
We construct an $h$-local acceptance family $(\{S_\sigma\}_{\sigma\in \Sigma},U)$ with $|U|\le|\Sigma|^{O(h)}$.
We call a subset $Y\subseteq \Sigma$ accepting if it contains some $Z\in\Acal$, and rejecting otherwise. No rejecting subset contains an accepting subset.
Let the universe consist of all rejecting subsets $Y\subseteq \Sigma$ of size at most $h$, and let $S_\sigma$ contain all rejecting subsets that do not contain $\sigma$:
$$
    U:=\{Y\subseteq \Sigma:|Y|\le h,\ Z\nsubseteq Y \text{ for every }Z\in\Acal\},
    \quad
    S_\sigma:=\{Y\in U:\sigma\notin Y\}.
$$
For every $X\subseteq \Sigma$ with $|X|\le h$, the union of the sets $S_\sigma$ for $\sigma\in X$ consists of all rejecting subsets that do not contain $X$:
$$
    \bigcup_{\sigma\in X}S_\sigma=\{Y\in U:X\nsubseteq Y\}.
$$
If $X$ is accepting, then no rejecting $Y$ contains $X$, so all of $U$ is covered. Otherwise, $X$ is rejecting and $|X|\le h$, so $X\in U$ is itself an uncovered universe element. This proves the required property.

\paragraph{$\Wone$-hardness of $k$-\setcover.}
The construction above gives universe size $|U^r|\le p^{O(h)}$ for each test, so the union over all tests $r\in[m]$ has size $|U|\le mp^{O(h)}$.
In the reduction of~\cite{DBLP:journals/jacm/SLM19}, $m=n^{O(1)}$ and $p=(\log n)^{k^{O(1)}}$. Hence taking $h=\frac{\log n}{k^{O(1)}\log\log n}$, with a suitable exponent in $k^{O(1)}$, keeps the reduction polynomial-time and gives the near-logarithmic $\Wone$-hardness of $k$-\setcover.

To obtain a smaller universe, we give a reduction from \vectorsum to \labelcover. It constructs a linear hash family that hashes the vectors into a smaller space, and for each hash function, a \labelcover test checks whether the hashed instance has a solution.
Combining this construction with the reduction from \labelcover to $k$-\setcover gives \cref{thm:intro-small-universe-w1}.
In particular, we obtain $\frac{\log u}{k^{O(1)}\log\log u}$-inapproximability of $k$-\setcover even when $u=(\log n)^{1.01}$.

\subsubsection{Self-Reductions for $k$-\setcover}
We introduce a reduction from small-universe \setcover to large-gap \setcover, as well as its reverse reduction.

\paragraph{The \mmsalevel{} framework.}
We first recall the framework of \cite{GLRZ26}, which represents \labelcover and \setcover instances as monotone circuits.

For a monotone circuit $\Phi:\{0,1\}^A\to\{0,1\}$, let $\opt(\Phi):=\min_{x:\Phi(x)=1}\|x\|_0$ denote the minimum Hamming weight of a satisfying assignment.
The $k$-Minimum Monotone Satisfying Assignment ($k$-\mmsalevel{}) problem asks whether there is a satisfying assignment of $\Phi$ with Hamming weight at most $k$, that is, $\opt(\Phi)\le k$. For $h\ge k$, its $(k,h)$-gap version requires us to distinguish between $\opt(\Phi)\le k$ and $\opt(\Phi)>h$.

The $k$-\setcover problem can be viewed as an \mmsalevel{} problem on AND--OR circuits, which we call \mmsalevel{2}. Consider a $k$-\setcover instance $I=(\{S_a\}_{a\in A}, U)$. For $a\in A$, let the variable $x_a$ indicate whether to select the set $S_a$. Let
$$
    \Psi(x):=\AND{e\in U}\ \OR{a:\,e\in S_a} x_a.
$$
Then $\Psi(x)=1$ if and only if the sets $S_a$ with $x_a=1$ cover $U$, and hence $\opt(I)=\opt(\Psi)$.
Similarly, we view \labelcover as \mmsalevel{} on AND--OR--AND--OR circuits, which we call \mmsalevel{4}.

The reductions in \cite{DBLP:conf/icalp/Lin19,GLRZ26} take as input a $k$-\setcover instance with a small universe, and output a large-gap $k$-\setcover instance. Their reductions has two steps: a gap-producing reduction from \mmsalevel{2} to \mmsalevel{4}, and a gap-preserving reduction from \mmsalevel{4} back to \mmsalevel{2}.
\begin{itemize}
    \item \emph{From $\text{MMSA}_2$ to $\text{MMSA}_4$.}
    \cite{GLRZ26} uses a hash family which hashes the set labels $A$ to a small alphabet $\Sigma$ with $|\Sigma|=h^{O(1)}$.
    For each hash test, we enumerate all $t:=|\Sigma|^k$ possibilities for the hash values of the $k$ chosen set labels, and check whether the selected sets cover $U$ for each guess. This can be expressed using an \mmsalevel{4} formula of the following form:
    $$
        \Phi(x):=\AND{r\in[m]}\OR{b\in[t]}\AND{e\in U}\OR{a\in A:\\b\in P_{r,a},\,e\in S_a}x_a.
    $$
    Here $P_{r,a}\subseteq[t]$ consists of all guesses $b$ that give the correct hash value for label $a$ under test $r$.
    We use $\Phi_r$ to denote the $r$-th term of $\Phi$ for every $r\in[m]$, and call each $\Phi_r$ a test.

    \item \emph{From $\text{MMSA}_4$ to $\text{MMSA}_2$.}
    \cite{DBLP:conf/icalp/Lin19} uses hypercube set systems to reduce $k$-\labelcover to $k$-\setcover. In \mmsalevel{} language, the process of applying a hypercube set system is exactly swapping $\OR{b\in[t]}$ and $\AND{e\in U}$ by the distributive law. This changes the layer order from AND--OR--AND--OR to AND--AND--OR--OR, and merging adjacent layers of the same type gives an \mmsalevel{2} circuit.
    Each test contributes $|U|^t=|U|^{h^{O(k)}}$ clauses, hence the resulting $k$-\setcover instance has universe size $m|U|^{h^{O(k)}}$.
\end{itemize}

\paragraph{An algebraic interpretation of acceptance families.}
We view the acceptance family construction as a transformation from a monotone circuit $\Phi:\{0,1\}^{\Sigma}\to\{0,1\}$ to an \mmsalevel{2} circuit that agrees with $\Phi$ on assignments of weight at most $h$.
The idea is to enumerate all candidate assignments $z$ of weight at most $h$ and precompute $\Phi(z)$. Specifically, for $\|x\|_0\le h$, we have
$$
    \Phi(x)
    =\AND{z\in\{0,1\}^{\Sigma}:\\ \|z\|_0\le h,\,z\ge x}\Phi(z)
    =\AND{z\in\{0,1\}^\Sigma:\\ \|z\|_0\le h}\bigl(\Phi(z)\vee[z\not\ge x]\bigr)
    =\AND{z\in\{0,1\}^\Sigma:\\\Phi(z)=0,\, \|z\|_0\le h}\,\OR{\sigma\in \Sigma:\\z_\sigma=0}x_\sigma.
$$
Here $[z\not\ge x]$ denotes the Boolean value indicating whether $z\not\ge x$.
Since each $\Phi(z)$ is a precomputed constant, the factors with $\Phi(z)=1$ disappear, and the remaining factors are OR clauses. The resulting circuit is therefore \mmsalevel{2}.
If $\Phi(x)$ expresses ``$Z\subseteq X$ for some $Z\in\Acal$'', where $x$ is the indicator vector of $X$, then the set family corresponding to the resulting \mmsalevel{2} circuit is exactly the acceptance family constructed above.

Generally, for $\Lcal\subseteq \{0,1\}^\Sigma$, we use the $\Lcal$-local acceptance family in the following algebraic form in the remaining reductions:
\begin{equation}\label{eq:intro-acceptance-family}
    \Phi(x)
    =\AND{z\in\{0,1\}^\Sigma:\\ z\in\Lcal}\bigl(\Phi(z)\vee[z\not\ge x]\bigr),\qquad
    \text{for every }x\in\Lcal.
\end{equation}
For a function $F(x,y)$ that is monotone in $y$, we may apply the same identity to $y$ while keeping $x$ fixed. In this case, $F(x,z)$ remains a function of $x$, and the resulting formula is not necessarily \mmsalevel{2}.

\paragraph{From small universe to large gap.}
Our reduction follows the \mmsalevel{} framework of~\cite{GLRZ26}: we reduce \mmsalevel{2} to \mmsalevel{4} and then back to \mmsalevel{2}. The intermediate \mmsalevel{4} instance has the form
$$
    \Phi(x):=\AND{r\in[m]}\Phi_r(x),\qquad
    \Phi_r(x):=\OR{b\in[t]}\AND{e\in U}\OR{a\in A:\\b\in P_{r,a},\,e\in S_a}x_a.
$$

We first introduce the gap-preserving reduction from the \mmsalevel{4} instance $\Phi$ to \mmsalevel{2}.
For a fixed test $r$, we again partition the variables by putting $x_a$ and $x_{a'}$ in the same group if and only if $P_{r,a}=P_{r,a'}$.
Suppose there are $p$ groups, and let $y\in\{0,1\}^p$ record which groups contain a selected variable: $y_j=1$ if and only if at least one variable in the $j$-th group is set to $1$. Thus $\|y\|_0\le\|x\|_0$.
Note that $y$ is a function of $x$ depending on $r$.

Once $y$ is known, the conditions $b\in P_{r,a}$ allow us to discard some indices $b$ for which the $b$-th term of $\Phi_r$ cannot be $1$.
For each test, we compute a set $\Gamma(y)\subseteq[t]$ consisting of all indices $b$ whose terms may still be $1$ given $y$. The set $\Gamma(y)$ grows monotonically with $y$. Since every discarded term is $0$, we only need to enumerate $b\in\Gamma(y)$:
$$
    \Phi_r(x)=\OR{b\in\Gamma(y)}\AND{e\in U}\OR{a\in A:\\b\in P_{r,a},\,e\in S_a}x_a.
$$
We regard the right-hand side as a function $F_r(x,y)$.
Applying \cref{eq:intro-acceptance-family} to $F_r(x,y)$ with respect to $y$ gives, for every $\|x\|_0\le h$ (hence $\|y\|_0\le h$),
$$\begin{aligned}
    \Phi_r(x)&=F_r(x,y)
    =\AND{z\in\{0,1\}^p:\\\|z\|_0\le h}\bigl(F_r(x,z)\vee[z\not\ge y]\bigr)\\
    &=\AND{z\in\{0,1\}^p:\\\|z\|_0\le h}\bigg(
        \OR{b\in\Gamma(z)}\AND{e\in U}
        \OR{a\in A:\\b\in P_{r,a},\,e\in S_a}x_a
        \vee[z\not\ge y]\bigg).
\end{aligned}$$
For each $z$, we precompute $\Gamma(z)$ and swap $\OR{b\in\Gamma(z)}$ and $\AND{e\in U}$ to obtain $|U|^{|\Gamma(z)|}$ OR clauses.
To control this number, we restrict the enumeration to $z$ with $|\Gamma(z)|\le c:=h/k$.
To preserve soundness, we therefore need the gap-producing reduction to ensure that, in the NO case, every $x$ with $\|x\|_0\le h$ has a test $r$ satisfying both $\Phi_r(x)=0$ and $|\Gamma(y)|\le c$.
The resulting \mmsalevel{2} formula has at most $p^h|U|^{h/k}$ clauses per test.

We now turn to the gap-producing reduction from \mmsalevel{2} to \mmsalevel{4}, which must also guarantee $|\Gamma(y)|\le c$ at a rejecting test.
Similar to the reduction from \vectorsum to \labelcover, we use linear hashing to produce the gap.
To obtain the additional bound on $|\Gamma(y)|$, we encode the labels of the $k$ chosen sets using a list-recoverable code $C:\F_q^k\to\F_q^\ell$ with $\ell=O(k)$.
The code guarantees that any lists $Y_1,\ldots,Y_\ell\subseteq\F_q$ of total size at most $h$ contain at most $c$ codewords:
$$
    \Big|\Big\{v\in\F_q^k:
        C(v)_j\in Y_j\text{ for every }j\in[\ell]\Big\}\Big|
    \le c.
$$
We introduce additional variables for the codeword coordinates: for each coordinate $j$ and symbol $\beta\in\F_q$, a variable indicates whether $\beta$ is selected at coordinate $j$.
For any selection of at most $h$ variables, let $Y_j$ be the set of codeword symbols selected at coordinate $j$. Then $\sum_j|Y_j|\le h$.
The hash tests check the consistency between the chosen set labels and the codewords.
In the NO case, hashing provides a rejecting test, at which every $b\in\Gamma(y)$ records the hash values of set labels and their corresponding codeword in $Y_1\times\cdots\times Y_\ell$.
Since each codeword determines only one candidate $b$, list recoverability gives $|\Gamma(y)|\le c$ at this test.

\paragraph{From large gap to small universe.}
Our reverse reduction follows the element-sampling argument of \cite{DIMV14}.
Given a $\gamma$-gap instance, we repeatedly sample uncovered elements and query exact $k$-\setcover on the samples.
A cover of a sample usually covers most of the remaining universe.
With constant success probability, repeating this step either finds a sample with no $k$-cover, certifying $\opt>k$, or produces a global cover using at most $\gamma k$ sets.
All exact queries have universe size at most $|U|^{O(1/\gamma)}$.
Together, the two reductions connect small universes and large gaps under randomized $\FPT$ Turing reductions.

\subsection{Open Problems}
\label{subsec:open-problems}

Below we list several open problems.

\paragraph{Tight approximation ratio.}
We prove $\Wone$-hardness for an approximation factor of $\Omega_k\big(\frac{\log n}{\log\log n}\big)$, while the greedy algorithm achieves an $O(\log n)$ approximation.
Thus a $\log\log n$ factor remains between the bounds.
Is $\frac{\log n}{f(k)}$-gap $k$-\setcover $\Wone$-hard for some computable function $f$?

\paragraph{Hardness at logarithmic primal VC dimension.}
The standard bound $d^*<2^{d+1}$, together with the $O(d^*\log k)$-approximation algorithm~\cite{ChekuriCH09}, gives a randomized polynomial-time $O(2^d\log k)$-approximation for $k$-\setcover with primal VC dimension $d$.
Our result \cref{thm:intro-vc-hardness} only rules out $O_k\big(\frac{d}{\log d}\big)$-approximation algorithms.
This motivates seeking $\gamma$-gap hardness with primal VC dimension logarithmic in $\gamma$.
Is $\gamma$-gap $k$-\setcover $\Wone$-hard even on instances with primal VC dimension $d=O_k(\log\gamma)$?

\paragraph{Between $\Wone$ and $\Wtwo$.}
Exact $k$-\setcover belongs to $\Wone$ when $|U|\le f(k)\log n$ \cite{DBLP:conf/soda/LinRSW23}, whereas the problem is $\Wtwo$-complete when $|U|=\Omega(n)$.
Our universe--gap equivalence further connects the complexity between subpolynomial universes and superconstant-factor approximation under randomized $\FPT$ Turing reductions.
For which $u=u(n)$ satisfying $u=\omega(\log n)$ and $u=n^{O(1)}$ does exact $k$-\setcover with $|\Scal|=n$ and $|U|\le u$ belong to $\Wone$, and for which is it $\Wtwo$-hard?
In particular, does the problem belong to $\Wone$, or is it $\Wtwo$-hard, when $u=O((\log n)^2)$?

\subsection{Organization}

\Cref{sec:preliminaries} defines the problems and complexity hypotheses and states the known \vectorsum hardness bounds.
\Cref{sec:klm} introduces acceptance families and hypercube set systems, and gives the reductions through \labelcover and proves the $\Wone$-hardness results.
\Cref{sec:tradeoff} develops the self-reductions and the universe--gap tradeoff, and proves the $k$-\ov running-time lower bound.
\Cref{sec:vc-dimensions} analyzes the primal and dual VC dimension of the \vectorsum-to-\setcover reduction.
The constructions of the linear separators and the list recoverable code are in the appendices.

    \section{Preliminaries}
\label{sec:preliminaries}

For a condition $E$, let $[E]\in\{0,1\}$ denote its Boolean indicator, which equals to $1$ if $E$ holds and $0$ otherwise.
We use the coordinatewise partial order on Boolean vectors: for $x,y\in\{0,1\}^n$, we write $x\ge y$ if $x_i\ge y_i$ for every $i\in [n]$, and $x\not\ge y$ otherwise.

\subsection{Complexity Hypotheses}

A parameterized problem with input $I$ and parameter $k$ is fixed-parameter
tractable ($\FPT$) if it has an algorithm running in
$f(k)|I|^{O(1)}$ time for some computable function $f$.
We use the following hypotheses.

\begin{hypothesis}[$\Wone\ne\FPT$]
    The $k$-\clique problem is not fixed-parameter tractable.
\end{hypothesis}

We use the randomized version of the logarithmic-dimension form of the $k$-\ov Hypothesis from \cite[Hypothesis~3.6]{DrabikDFMW25}.

\begin{hypothesis}[Randomized $k$-\ov Hypothesis]
\label{hyp:kov}
    For every fixed integer $k\ge2$ and constant $\varepsilon>0$, there is a constant $c>0$ such that no randomized algorithm solves the following problem with success probability at least $2/3$ in $O(n^{k-\varepsilon})$ time:
    given $k$ sets $X_1,\ldots,X_k\subseteq\{0,1\}^d$ with $|X_i|=n$ and $d\le c\log n$, decide whether there are $a_i\in X_i$ for $i\in[k]$ such that
    $$
        \prod_{i=1}^k a_i[j]=0\qquad\text{for every }j\in[d].
    $$
\end{hypothesis}

\subsection{Parameterized Problems}

\subsubsection{Minimum Monotone Satisfying Assignment}
In our paper, we only consider $k$-\mmsalevel{t} problems where $t$ is an even integer.
\begin{definition}[$k$-\mmsalevel{t}]
\label{def:mmsa}
    Let $n,m_1,\dots,m_{t-1}\ge 1$ be integers, and let $N_{i}\subseteq[n]$ for every $i\in[m_1]\times\dots\times[m_{t-1}]$.
    An \mmsalevel{t} circuit $\Phi:\{0,1\}^n\to\{0,1\}$ is defined as
    $$
        \Phi(x):=\AND{i_1=1}^{m_1}\OR{i_2=1}^{m_2}\cdots\AND{i_{t-1}=1}^{m_{t-1}}\OR{a\in N_{i_1,\ldots,i_{t-1}}}x_a.
    $$
    Define
    $$
        \opt(\Phi):=\min_{x\in\{0,1\}^n:\Phi(x)=1}\|x\|_0.
    $$
    The $k$-\mmsalevel{t} problem asks whether $\opt(\Phi)\le k$.
    For $h\ge k$, its $(k,h)$-gap version asks us to distinguish between $\opt(\Phi)\le k$ and $\opt(\Phi)>h$.
\end{definition}

\subsubsection{$k$-\setcover as \mmsalevel{2}}
A $k$-\setcover instance $I=(\Scal,U)$ consists of a universe $U$ and a family $\Scal$ of subsets of $U$ with $|\Scal|=n$. Let $\opt(I)$ denote the minimum size of a subfamily $\Ccal\subseteq\Scal$ such that $\bigcup_{S\in\Ccal}S=U$. The $k$-\setcover problem asks whether $\opt(I)\le k$.

We formulate $k$-\setcover and its gap version as \mmsalevel{2} problems.

\begin{definition}[$k$-\setcover]
\label{def:setcover}
    Let $k,n\ge1$ be integers, let $U$ be a universe of size $n^{O(1)}$, and let $N_e\subseteq[n]$ for every $e\in U$. For $x\in\{0,1\}^n$, define
    $$
        \Psi(x):=\AND{e\in U}\OR{a\in N_e}x_a.
    $$
    The $k$-\setcover problem asks whether $\opt(\Psi)\le k$.
    For $h\ge k$, its $(k,h)$-gap version distinguishes between $\opt(\Psi)\le k$ and $\opt(\Psi)>h$.
    For $\gamma\ge1$, we use $\gamma$-gap $k$-\setcover to denote the $(k,\gamma k)$-gap problem.
\end{definition}

Here $N_e=\{a\in[n]:e\in S_a\}$ is the set of indices of sets that cover $e\in U$, and $x_a$ indicates whether $S_a$ is selected.
Thus $\Psi(x)=1$ precisely when the selected sets with $x_a=1$ cover $U$, so the two definitions of $k$-\setcover are equivalent.

\subsubsection{$k$-\labelcover as \mmsalevel{4}}
A $k$-\labelcover instance $I$ consists of a bipartite graph with left vertex set $A=A_1\union\cdots\union A_k$ and right vertex set $B=B_1\union\cdots\union B_m$, where each $A_i$ is indexed by $[n]$ and each $B_r$ by $[t]$.
The set $N_{r,b,i}$ consists of the indices of the neighbors of $b\in B_r$ in $A_i$.
A selection $X\subseteq A$ passes test $r\in[m]$ if some $b\in B_r$ has a neighbor in $X\cap A_i$ for every $i\in[k]$.
Let $\opt(I)$ denote the minimum size of a selection $X$ that passes every test.

We formulate $k$-\labelcover and its gap version as \mmsalevel{4} problems.

\begin{definition}[$k$-\labelcover]
\label{def:labelcover}
    Let $k,n,m,t\ge 1$ be integers, and let $N_{r,b,i}\subseteq[n]$ for every $r\in[m]$, $b\in[t]$, and $i\in[k]$.
    For $x\in\{0,1\}^{[k]\times[n]}$, define
    $$
        \Phi(x):=\AND{r\in[m]}\OR{b\in[t]}\AND{i\in[k]}\OR{a\in N_{r,b,i}}x_{i,a}.
    $$
    The $k$-\labelcover problem asks whether $\opt(\Phi)\le k$.
    For $h\ge k$, its $(k,h)$-gap version distinguishes between $\opt(\Phi)\le k$ and $\opt(\Phi)>h$.
\end{definition}

We use
$$
    \Phi_r(x):=\OR{b\in[t]}\AND{i\in[k]}\OR{a\in N_{r,b,i}}x_{i,a}
$$
to denote the $r$-th test of $\Phi$, so $\Phi(x)=\AND{r\in[m]}\Phi_r(x)$.
The variable $x_{i,a}$ indicates whether vertex $a\in A_i$ is selected.
Thus $\Phi_r(x)=1$ precisely when the selected vertices $X=\bigcup_{i\in[k]}\{a\in A_i: x_{i,a}=1\}$ pass test $r$, and $\Phi(x)=1$ precisely when they pass every test.
This proves the equivalence of the two definitions of \labelcover.

For a \labelcover formula $\Phi$, we often use
$$
    P_{r,i,a}:=\{b\in[t]:a\in N_{r,b,i}\}
$$
to denote the indices of the neighbors of $a\in A_i$ in $B_r$, and
$$
    p:=\max_{r\in[m],i\in[k]}|\{P_{r,i,a}:a\in[n]\}|
$$
to denote the maximum number of distinct neighborhoods in $B_r$ among vertices in $A_i$.

\subsubsection{$k$-\vectorsum}

\begin{definition}[$k$-\vectorsum] \label{def:k-vector-sum}
\label{def:vectorsum}
    Given $k$ sets $X_1,\ldots,X_k\subseteq\F_2^D$ with $|X_i|=n$ for every $i\in[k]$, the $k$-\vectorsum problem asks whether there are $v_i\in X_i$ for $i\in[k]$ such that $\sum_{i=1}^k v_i=\vec0$.
\end{definition}

We also allow the groups $X_i$ to have different sizes and write $n:=\max_{i\in[k]}|X_i|$.
To make all groups have size $n$, embed each original vector $v\in X_i$ as $(v,\vec0)\in\F_2^{D+k}$ and pad $X_i$ with distinct dummy vectors $(w,e_i)$, where $w\in\F_2^D$ and $e_i$ is the $i$-th unit vector in $\F_2^k$.
Any tuple containing a dummy vector from $X_i$ has its $i$-th appended coordinate equal to $1$, so it cannot sum to $\vec0$.
Thus the padding preserves the existence of a zero-sum tuple and increases the dimension to $D+k$.
We may therefore assume that $|X_i|=n$ for every $i\in[k]$.

We use the following \Wone-hardness result for \vectorsum.
\begin{lemma}[{\cite{DBLP:conf/esa/BhattacharyyaGGS16}}] \label{lem:vectorsum-w1}
    The $k$-\vectorsum problem is $\Wone$-hard even when $D=O(k\log n)$.
\end{lemma}
The formulation in \cite{DBLP:conf/esa/BhattacharyyaGGS16} includes a target vector $t\in\F_2^D$: given groups $Y_1,\ldots,Y_k\subseteq\F_2^D$, it asks whether there are $y_i\in Y_i$ with $\sum_{i=1}^k y_i=t$.
To obtain a zero-sum instance, set $X_1:=\{y+t:y\in Y_1\}$ and $X_i:=Y_i$ for $i>1$.
This translation preserves the group sizes and dimension, and $\sum_i x_i=\vec 0$ if and only if $\sum_i y_i=t$.
Hence these two formulations are equivalent.

\subsection{VC Dimension}
\label{subsec:vc-basics}

For a finite set system $(\Scal,U)$, a subset $X\subseteq U$ is \emph{shattered} if $\{S\cap X:S\in\Scal\}=2^X$.
Its \emph{VC dimension} $d$ is the largest size of a shattered subset.
Define the \emph{dual set system} of $(\Scal,U)$ as $(\Scal^*,\Scal)$, where
$$
    \Scal^*:=\{\{S\in\Scal:e\in S\}:e\in U\}.
$$
The \emph{dual VC dimension} $d^*$ of $(\Scal,U)$ is the VC dimension of $(\Scal^*,\Scal)$.
The \emph{shatter function} of $\Scal$ is
$$
    \pi_{\Scal}(s):=
    \max_{X\subseteq U:\,|X|\le s}|\{S\cap X:S\in\Scal\}|.
$$

\begin{theorem}[Sauer's lemma]\label{thm:vc-sauer}
    If $(\Scal,U)$ has VC dimension at most $d\ge1$, then,
    for every $s\ge d$,
    $$
        \pi_{\Scal}(s)\le\Big(\frac{\e s}{d}\Big)^d.
    $$
\end{theorem}

\begin{lemma}\label{lem:vc-union}
    Let $\Scal_1,\ldots,\Scal_m$ be set families on the same
    universe, each with VC dimension at most $d$.
    Their union has VC dimension $O(d+\log m)$.
\end{lemma}
\begin{proof}
    Let $s$ be the VC dimension of $\bigcup_{r=1}^m\Scal_r$. Suppose $s\ge 4d$. By \cref{thm:vc-sauer}, we have
    $$
        2^s\le\sum_{r=1}^m\pi_{\Scal_r}(s)\le m\Big(\frac{\e s}{d}\Big)^d.
    $$
    Taking logarithms and using $\log_2\frac{s}{d}\le \frac{s}{2d}$ gives $s\le 2\log m+4 d$.
    Thus, either $s\le 4d$ or $s\le 2\log m+4 d$. In both cases, $s=O(d+\log m)$.
\end{proof}

\begin{lemma}\label{lem:vc-intersections}
    For an integer $d\ge0$, let $\Scal$ be a set family such that
    $|S\cap T|\le d$ for any distinct $S,T\in\Scal$.
    Then $\Scal$ has VC dimension at most $d+1$.
\end{lemma}
\begin{proof}
    Assume for the sake of contradiction that there exists a shattered set $X$ of size $d+2$, and pick any $x\in X$.
    There are $S,T\in\Scal$ with $S\cap X=X$ and $T\cap X=X\setminus\{x\}$.
    Then $S\neq T$ and $X\setminus\{x\}\subseteq S\cap T$, so $|S\cap T|\ge d+1$, a contradiction.
\end{proof}

    \section{Near-Logarithmic Inapproximability of \ensuremath{k}-\setcover}
\label{sec:klm}

\subsection{Algebraic Tools for Monotone Circuits}

We first formulate local acceptance families and hypercube set systems as algebraic transformations for our reductions from \mmsalevel{4} to \mmsalevel{2}.

The acceptance family construction uses monotonicity and absorption. We enumerate all $z\in\Lcal$ and precompute $\Phi(z)$ as constants, rather than include $\Phi(x)$ as a subcircuit. On input $x$, we take the conjunction of all $\Phi(z)$ with $z\ge x$.
The resulting circuit agrees with $\Phi(x)$ on every $x\in\Lcal$.

\begin{lemma}[$\Lcal$-local acceptance family] \label{lem:local-acceptance-general}
    Given $n\in\N$, $\Lcal\subseteq\{0,1\}^{n}$ and a monotone circuit $\Phi:\{0,1\}^{n}\to\{0,1\}$, for every $x\in\{0,1\}^n$, we have
    $$
        \Phi(x) \le \AND{z\in\Lcal}\bigg(\Phi(z)\vee \OR{a\in [n]:\, z_a=0}x_a\bigg).
    $$
    The equality holds when $x\in\Lcal$.
\end{lemma}
\begin{proof}
    By monotonicity, we have $\Phi(z)\ge \Phi(x)$ for every $z\ge x$. Hence
    $$
        \Phi(x)\le \AND{z\in\Lcal:\,z\ge x} \Phi(z)
        = \AND{z\in\Lcal} \big( \Phi(z)\lor[z\not\ge x] \big)
        = \AND{z\in\Lcal}\bigg(\Phi(z)\vee \OR{a\in [n]:\, z_a=0}x_a\bigg).
    $$
    When $x\in\Lcal$, we have $\AND{z\in\Lcal:\,z\ge x}\Phi(z)\le \Phi(x)$, hence the equality holds.
\end{proof}

After omitting the factors with $\Phi(z)=1$, the right-hand side is an \mmsalevel{2} formula with at most $|\Lcal|$ clauses. Taking $\Lcal:=\{z\in\{0,1\}^n:\|z\|_0\le h\}$ gives the following corollary.

\begin{corollary}[$h$-local acceptance family] \label{lem:local-acceptance}
    Given $n\in\N$, $h\ge0$ and a monotone circuit $\Phi:\{0,1\}^n\to\{0,1\}$, for every $x\in\{0,1\}^n$, we have
    $$
        \Phi(x)\le \AND{z\in\{0,1\}^n:\, \|z\|_0\le h}\bigg(\Phi(z)\vee \OR{a\in [n]:\, z_a=0}x_a\bigg).
    $$
    The equality holds when $\|x\|_0\le h$.
\end{corollary}

The hypercube set system follows directly from distributivity.
\begin{lemma}[Hypercube set system]
    Let $B,U$ be finite sets. Given any Boolean circuits $C_{b,e}:\{0,1\}^n\to\{0,1\}$ for $b\in B,e\in U$, the following equality holds for every $x\in\{0,1\}^n$:
    $$
        \bigvee_{b\in B}\bigwedge_{e\in U}C_{b,e}(x)
        = \bigwedge_{f:B\to U}\bigvee_{b\in B}C_{b,f(b)}(x).
    $$
\end{lemma}

If every $C_{b,e}$ is a disjunction of input variables, the right-hand side is an \mmsalevel{2} formula with $|U|^{|B|}$ clauses.

\subsection{From \labelcover to $k$-\setcover}

\begin{theorem} \label{thm:labelcover-to-setcover}
    There is an algorithm that takes as input an integer $k\ge2$ and $h\ge k$, and a \labelcover formula $\Phi$ on $x\in\{0,1\}^{[k]\times [n]}$ given by $P_{r,i,a}\subseteq[t]$:
    $$
        \Phi(x)=\AND{r\in[m]}\Phi_r(x), \qquad
        \Phi_r(x)=\OR{b\in [t]}\AND{i\in [k]}\OR{a\in[n]:\\ b\in P_{r,i,a}}x_{i,a},
    $$
    and outputs an \mmsalevel{2} formula $\Psi:\{0,1\}^{[k]\times[n]}\to\{0,1\}$, such that
    \begin{itemize}
        \item \emph{(Size.)} The AND gate of $\Psi$ has fan-in at most $m(kp)^{h}$, where
        $p:=\max_{r\in[m],i\in[k]}|\{P_{r,i,a}:a\in[n]\}|$.
        \item \emph{(Completeness.)} If $\opt(\Phi)\le k$, then $\opt(\Psi)\le k$.
        \item \emph{(Soundness.)} If $\opt(\Phi)>h$, then $\opt(\Psi)>h$.
        \item \emph{(Running time.)} The reduction runs in time $O(nmt(kp)^{h})$.
    \end{itemize}
\end{theorem}
\begin{proof}
    Fix a test $r\in[m]$.
    To apply the local acceptance family to at most $kp$ variables, let $\Pcal_i^r:=\{P_{r,i,a}:a\in[n]\}$ for $i\in[k]$, and let $\Pcal^r:=\{(i,P):i\in[k],P\in\Pcal_i^r\}$. Define $y^r_{i,P}:=\OR{a\in[n]:\,P_{r,i,a}=P}x_{i,a}$ for every $(i,P)\in\Pcal^r$.
    We omit the superscript $^r$ for simplicity when $r$ is fixed.
    The variables $y_{i,P}$ record whether a class is selected, and satisfy $\|y\|_0\le\|x\|_0$.
    Then we have
    $$
        \Phi_r(x) = \OR{b\in[t]}\AND{i\in[k]}\OR{P\in\Pcal_i:\\ b\in P}\OR{a\in[n]:\\ P_{r,i,a}=P} x_{i,a}
        = \OR{b\in[t]}\AND{i\in[k]}\OR{P\in\Pcal_i:\\ b\in P} y_{i,P}
        =: F(y).
    $$
    Apply $h$-local acceptance family to $F$. For every $x$ with $\|x\|_0\le h$, since $\|y\|_0\le \|x\|_0\le h$, we have
    $$\begin{aligned}
        \Phi_r(x) &= \AND{z\in\{0,1\}^{\Pcal}:\\ \|z\|_0\le h} \bigg(
            F(z) \lor \OR{(i,P)\in\Pcal:\\ z_{i,P}=0} y_{i,P}
        \bigg) && \text{(by \cref{lem:local-acceptance})}\\
        &= \AND{z\in\{0,1\}^{\Pcal}:\\ \|z\|_0\le h,\, F(z)=0} \ \OR{(i,P)\in\Pcal:\\ z_{i,P}=0}\OR{a\in[n]:\\P_{r,i,a}=P}x_{i,a}
         && \text{(expand $y_{i,P}$)}\\
        &= \AND{z\in\{0,1\}^{\Pcal}:\\ \|z\|_0\le h,\, F(z)=0}\ \OR{i\in[k],a\in[n]:\\ z[i,P_{r,i,a}]=0} x_{i,a}
        =: \Psi_r(x).
    \end{aligned}$$
    Let $\Psi(x):=\AND{r\in[m]}\Psi_r(x)$, then $\Psi$ is an \mmsalevel{2} formula. Moreover, $\Phi(x)=\Psi(x)$ for every $x$ with $\|x\|_0\le h$, hence the completeness and soundness hold.

    For every test $r\in[m]$, we have $|\Pcal^r|=\sum_{i\in[k]}|\Pcal_i^r|\le kp$, so the number of $z\in\{0,1\}^{\Pcal^r}$ with $\|z\|_0\le h$ is at most $(kp)^h$. Thus $\Psi$ has at most $m(kp)^{h}$ clauses in total.
    Computing the families $\Pcal_i^r$ for all tests and groups takes $O(knmt)$ time.
    For each $z$, evaluating $F(z)$ and writing its corresponding clause take $O(knt)$ time.
    Summing over all supports gives total time $O(nmt(kp)^h)$.
\end{proof}

\subsection{From \vectorsum to \labelcover}

We introduce a class of hash families called linear separators. We linearly hash the vectors in the \vectorsum instance into a smaller space, and for each hash function, we use an OR-AND-OR circuit to check whether the hashed instance still has a zero-sum tuple.
Linearity preserves zero-sum tuples.
To reject every selection of at most $h$ vectors in the NO case, the family must contain a hash function that maps all nonzero sums of tuples from the selection to nonzero vectors.

We formally define linear separators as follows.
The construction of linear separators is deferred to \cref{app:separators}.

\begin{definition}[Linear separator]\label{def:linear-separator}
    For $D,d\in\N$~and $A\ge1$, a family $\Hcal$ of $\F_2$-linear maps $\psi:\F_2^D\to\F_2^d$ is a \emph{$(D,d,A)$-linear separator}
    if, for every $Z\subseteq\F_2^D\setminus\{\vec0\}$ with $|Z|\le A$,
    some $\psi\in\Hcal$ satisfies $\psi(z)\ne\vec0$ for every $z\in Z$.
    Each linear map is represented by a matrix in $\F_2^{d\times D}$.
\end{definition}

\begin{lemma}\label{lem:linear-separator}
    For $D\in\N$~and $A\ge1$, there is a $(D,d,A)$-linear separator $\Hcal$ with $2^d=O(A)$ and $|\Hcal|=O(DA^2)$,
    constructible deterministically in $(DA)^{O(1)}$ time.
\end{lemma}

Now we show the reduction from \vectorsum to \labelcover using linear separators.

\begin{theorem}\label{thm:vectorsum-to-labelcover}
    Given an integer $k\ge2$ and $h\ge k$, a $k$-\vectorsum instance $I=(X_1,\ldots,X_k)$ with $X_i\subseteq\F_2^D$ and $|X_i|=n$, and a $(D,d,(h/k)^k)$-linear separator $\Hcal$,
    there is a deterministic algorithm constructing a \labelcover formula
    $$
        \Phi(x)=\AND{r\in[m]}\OR{b\in [t]}\AND{i\in [k]}\OR{a\in[n]:\\ b\in P_{r,i,a}}x_{i,a},
    $$
    on $x\in\{0,1\}^{[k]\times[n]}$ with the following properties, where $q:=2^d$:
    \begin{itemize}
        \item \emph{(Size.)} $m=|\Hcal|$, $t=q^{k-1}$, and $p\le q$, where $p:=\max_{r\in[m],i\in[k]}|\{P_{r,i,a}:a\in[n]\}|$.
        \item \emph{(Completeness.)} If $I$ is a YES instance, then $\opt(\Phi)\le k$.
        \item \emph{(Soundness.)} If $I$ is a NO instance, then $\opt(\Phi)>h$.
        \item \emph{(Running time.)} The construction takes $(nmDq^k)^{O(1)}$ time.
    \end{itemize}
\end{theorem}
\begin{proof}
    We use the linear separator $\Hcal=\{\psi_r\}_{r\in[m]}$ to preserve zero-sum tuples and reject every selection of at most $h$ vectors in a NO instance.
    Let $\Sigma:=\F_2^d$ and $B:=\{b\in\Sigma^k:\sum_i b_i=\vec0\}$.
    Denote the vectors in $X_i$ as $v_{i,1},\dots,v_{i,n}$, and let $P_{r,i,a}:=\{b\in B:b_i=\psi_r(v_{i,a})\}$.
    Thus the resulting \mmsalevel{4} formula is
    $$
        \Phi(x)=
        \AND{r\in[m]}\OR{b\in\Sigma^k:\\\sum_i b_i=\vec 0}\,\AND{i\in[k]}\OR{a\in[n]:\\ \psi_r(v_{i,a})=b_i} x_{i,a}.
    $$
    Intuitively, the formula requires that, for every hash function $\psi_r$, there exist selected vectors $v_{1,a_1}\in X_1,\dots,v_{k,a_k}\in X_k$ such that $\sum_i\psi_r(v_{i,a_i})=\vec 0$.
    The variable $x_{i,a}$ indicates whether vector $v_{i,a}$ is selected.
    For each test $r$, $b\in B$ guesses one hash value per group whose sum is zero, and the inner gates check that each guessed value occurs among the selected vectors of its group.

    In the YES case, there exists $(v_{1,a_1},\dots,v_{k,a_k})$ such that $\sum_i v_{i,a_i}=\vec 0$.
    For every $r$, linearity gives $\sum_i\psi_r(v_{i,a_i})=\vec0$, hence setting variables $x_{1,a_1},\dots,x_{k,a_k}$ to $1$ and all others to $0$ satisfies $\Phi$, and thus $\opt(\Phi)\le k$.

    In the NO case, $\sum_{i}v_i\neq\vec 0$ for every $v_i\in X_i$. Fix any $x$ with $\|x\|_0\le h$, and let $L_i:=\{v_{i,a}:a\in[n],x_{i,a}=1\}$.
    Consider the set of sums of selected tuples:
    $$
        Z:=\bigg\{\sum_{i=1}^k v_i:v_i\in L_i\bigg\}.
    $$
    We have $\vec 0\notin Z$ since $I$ is a NO instance, and the AM--GM inequality gives $|Z|\le\prod_i|L_i|\le(h/k)^k$.
    By \cref{def:linear-separator}, there exists $r\in[m]$ such that $\psi_r(z)\neq\vec 0$ for every $z\in Z$.
    By linearity, we have $\sum_i\psi_r(v_i)\neq\vec 0$ for every $v_i\in L_i$.
    Hence $\Phi_r(x)=0$, proving soundness.

    By construction, $m=|\Hcal|$ and $|B|=2^{d(k-1)}=q^{k-1}$.
    For fixed $r$ and $i$, $P_{r,i,a}$ depends only on $\psi_r(v_{i,a})\in\Sigma$, hence $p\le q$.
    Computing $\psi_r$ and enumerating these neighborhoods takes $(nmDq^k)^{O(1)}$ time.
\end{proof}

\subsection{\Wone-Hardness of \ensuremath{k}-\setcover}

Composing \cref{thm:vectorsum-to-labelcover,thm:labelcover-to-setcover}, we have the following reduction from $k$-\vectorsum to $k$-\setcover, which gives the \Wone-hardness of approximating $k$-\setcover.

\begin{lemma}\label{lem:vectorsum-small-universe}
    There exists a constant $0<\alpha<1$ with the following property.
    Given a $k$-\vectorsum instance $I=(X_1,\ldots,X_k)$ with $X_i\subseteq\F_2^D$, $|X_i|=n$, and $D=O(k\log n)$, let $\rho,h$ satisfy
    $$
        k^{\Omega(k^2)}\le\rho\le n^{O(1)},
        \qquad
        k\le h\le\frac{\alpha\log\rho}{k\log\log\rho}.
    $$
    There is a deterministic algorithm that constructs a $(k,h)$-gap \setcover instance $(\Scal,U)$ with $|\Scal|=kn$ and $|U|\le\rho\log n$ in $n^{O(1)}$ time.
\end{lemma}
\begin{proof}
    Let $A:=(h/k)^k$ and construct a $(D,d,A)$-linear separator $\Hcal$ by \cref{lem:linear-separator}, with $q:=2^d=O(A)$ and $m:=|\Hcal|=O(DA^2)$.
    Applying \cref{thm:vectorsum-to-labelcover} to $I$ with separator $\Hcal$, we obtain a \labelcover instance with parameters $m=O(DA^2)$, $t=q^{k-1}$ and $p=O(A)$.
    Then apply \cref{thm:labelcover-to-setcover}, and write the output \mmsalevel{2} instance as a $k$-\setcover instance $(\Scal,U)$. The output instance satisfies $|\Scal|=kn$ and
    $$
        |U|\le m(kp)^h
        \le O(Dh^{2k})(kh^k)^{O(h)}
        \le 2^{O(kh\log h)}\log n
        \le 2^{O(\alpha\log\rho)}\log n.
    $$
    For sufficiently small constant $\alpha>0$, we have $|U|\le \rho\log n$.

    For the running time, we have $h^{k^2}\le2^{kh\log h}\le\rho^{O(1)}$ and $m(kp)^h\le \rho^{O(1)}\log n$.
    Thus constructing the linear separator and performing the two reductions take $(nDh^{k^2})^{O(1)}(kp)^h=n^{O(1)}$ time in total.
\end{proof}

\begin{theorem}[\cref{thm:intro-small-universe-w1} restated] \label{thm:small-universe-w1}
    There exists a constant $\alpha>0$ such that,
    for every \FPT-computable $\rho=\rho(k,N)$ satisfying $k^{\Omega(k^2)}\le\rho\le N^{O(1)}$,
    $$
        \frac{\alpha\log\rho}{k^2\log\log\rho}\text{-gap $k$-\setcover}
    $$
    on instances $(\Scal,U)$ with $|\Scal|=N$ and $|U|\le\rho\log N$ is $\Wone$-hard.
\end{theorem}
\begin{proof}
    We start from a $k$-\vectorsum instance in \cref{lem:vectorsum-w1} with $|X_i|=n$ and $D=O(k\log n)$.
    Let $N:=kn$ and $\rho:=\rho(k,N)$.
    Apply \cref{lem:vectorsum-small-universe} with $h:=\frac{\alpha\log\rho}{k\log\log\rho}$ for sufficiently small $\alpha$, and we obtain a $k$-\setcover instance with $|\Scal|=N$ and $|U|\le \rho\log n\le \rho\log N$ in $\FPT$ time. This proves the claimed hardness.
\end{proof}

    \section{Self-Reductions for \ensuremath{k}-\setcover and the Universe--Gap Tradeoff}
\label{sec:tradeoff}

\subsection{List Recoverable Codes}
We introduce an average version of list recoverable codes which will be used in the gap-amplification reduction.
Given an error correcting code $C\subseteq\F_q^\ell$ and lists $Y_1,\dots,Y_\ell\subseteq\F_q$, we say $y\in C$ agrees with $Y$ if $y_j\in Y_j$ for every $j\in[\ell]$.
Standard zero-error $(h,c)$-list recoverability~\cite{GuruswamiIndyk01} requires that when $|Y_i|\le h$ for every $i$, there are at most $c$ codewords $y$ that agree with $Y$, while our definition requires it to hold when $\sum_i|Y_i|\le h$.
The formal definition of $(h,c)$-list recoverability is as follows.

\begin{definition}[$(h,c)$-list recoverable code]
\label{def:list-recoverable-code}
    Let $q$ be a prime power, $k,\ell$ be positive integers, and $h,c\ge 1$. An error correcting code $C:\F_q^k\to\F_q^\ell$ is an \emph{$(h,c)$-list recoverable code} if, for every $Y_1,\ldots,Y_\ell\subseteq\F_q$ with $\sum_{j=1}^\ell|Y_j|\le h$, the number of $x\in\F_q^k$ such that $C(x)$ agrees with $Y$ is at most $c$, that is,
    $$
        \big|\big\{
            a\in\F_q^k:
            C(a)\in Y_1\times\dots\times Y_\ell
        \big\}\big|\le c.
    $$
\end{definition}

In our reduction, we additionally require the code to be additively separable, which means the code $C$ can be represented by $k$ functions $f_1,\dots,f_k$ such that  $C(x)=f_1(x_1)+\dots+f_k(x_k)$ for every input $x\in\F_q^k$.
We show a randomized construction of an additively separable $(h,\lfloor h/k\rfloor)$-list recoverable code as follows, the proof of which is deferred to \cref{app:encoding}.
\begin{lemma}\label{lem:additive-code}
    There exist constants $\lambda\ge 1$ and $\beta>0$ such that for $k\in\N$, $h\ge k$, and a power of two $q\ge4$ satisfying
    $$
        h\le\frac{\beta\log q}{k\log\log q},
    $$
    there is a randomized algorithm that constructs the tables of $k$ functions $f_i:\F_q\to\F_q^{\lambda k}$ in $(kq)^{O(1)}$ time,
    such that the code $C:\F_q^k\to\F_q^{\lambda k}$ defined by $C(x_1,\ldots,x_k):=\sum_{i\in[k]}f_i(x_i)$ is injective and $(h,\lfloor h/k\rfloor)$-list recoverable with probability at least $1-q^{-1}$.
\end{lemma}

\subsection{From Small Universe to Large Gap}
\label{subsec:gap-amplification}

The reduction has three steps: first encode every choice of $k$ sets as a zero-sum vector tuple using a list recoverable code, then construct an \mmsalevel{4} formula using linear hashing, and convert it to \mmsalevel{2} using acceptance families and the hypercube identity.

Fix a $k$-\setcover instance $I_0=(\Scal_0,U_0)$ with $n:=|\Scal_0|$ and $u_0:=|U_0|\le n$, a positive integer $\ell$, and $h,c\ge1$, and let $q$ be the smallest power of two with $q\ge|\Scal_0|$.
Let
$$
    C:\F_q^k\to\F_q^\ell,
    \qquad C(a_1,\dots,a_k)=\sum_{i=1}^k f_i(a_i)
$$
be an injective $(h,c)$-list recoverable code, specified by the tables of the $f_i$.
Let $K:=k+\ell$ and assume $h\ge K$.

\paragraph{Step 1: encoding sets as vectors.}
In this step, we use the code $C$ to construct $K$ groups of sets and vectors.
We use the first $k$ groups to represent the message coordinates and the remaining $\ell$ groups to represent the codeword coordinates. Each possible value in a group is associated with a vector and a subset of $U_0$.

Pad $\Scal_0$ with copies of $\emptyset$ to make $|\Scal_0|=q$, and index each element in $\Scal_0$ by $S_\alpha$ for $\alpha\in\F_q$.
We construct vectors $v_{i,\alpha}\in\F_q^\ell$ and their associated sets $S_{i,\alpha}\subseteq U_0$ for $i\in[K]$ and $\alpha\in\F_q$ as follows:
\begin{itemize}
    \item For each $i\in[k]$ and $\alpha\in\F_q$, let $S_{i,\alpha}:=S_\alpha$, and $v_{i,\alpha}:=f_i(\alpha)$.
    \item For each $j\in[\ell]$ and $\beta\in\F_q$, let $S_{k+j,\beta}:=\emptyset$, and $v_{k+j,\beta}:=-\beta e_j$, where $e_j\in\F_q^\ell$ satisfies $e_j[j]=1$ and $e_j[j']=0$ for all other coordinates $j'$.
\end{itemize}
Since $C$ is injective, each $f_i$ is injective by the additive separability of $C$, so the vectors within each group are distinct by construction.

\begin{claim}\label{clm:modular-encoding}
    The sets and vectors in Step 1 have the following properties:
    \begin{itemize}
        \item There is an $\alpha\in\F_q^K$ with $\sum_{i=1}^K v_{i,\alpha_i}=\vec0$ and $\bigcup_{i=1}^K S_{i,\alpha_i}=U_0$ if and only if $\opt(I_0)\le k$.
        \item For every $L_1,\dots,L_K\subseteq\F_q$ with $\sum_{i=1}^K|L_i|\le h$, the number of zero-sum tuples of selected vectors is at most $c$:
        $$
            \bigg|\bigg\{\alpha\in L_1\times\dots\times L_K:
                \sum_{i=1}^K v_{i,\alpha_i}=\vec0\bigg\}\bigg|\le c.
        $$
    \end{itemize}
\end{claim}

\begin{proof}
    By construction, for every $(\alpha_1,\dots,\alpha_k,\beta_1,\dots,\beta_\ell)\in\F_q^K$, it holds that
    $$
        \sum_{i=1}^k v_{i,\alpha_i} + \sum_{j=1}^\ell v_{k+j,\beta_j}=\vec 0
        \quad\text{if and only if}\quad
        (\beta_1,\dots,\beta_\ell)=C(\alpha_1,\dots,\alpha_k).
    $$
    In the YES case, there exist original sets $S_{\alpha_1},\dots,S_{\alpha_k}$ covering $U_0$.
    Let $(\beta_1,\dots,\beta_\ell):=C(\alpha_1,\dots,\alpha_k)$.
    Then $\sum_{i=1}^k v_{i,\alpha_i}+\sum_{j=1}^\ell v_{k+j,\beta_j}=\vec 0$, and $S_{1,\alpha_1},\dots,S_{k,\alpha_k},S_{k+1,\beta_1},\dots,S_{K,\beta_\ell}$ cover $U_0$.
    In the NO case, every $x\in\F_q^K$ satisfies $\bigcup_{i=1}^K S_{i,x_i}=\bigcup_{i=1}^k S_{x_i}\neq U_0$, since the auxiliary sets are empty and no $k$ original sets cover $U_0$.

    For the second property, fix $L_1,\dots,L_K\subseteq\F_q$ with total size $\sum_i|L_i|\le h$.
    If $\sum_{i=1}^K v_{i,x_i}=\vec 0$, then $(x_{k+1},\dots,x_K)=C(x_1,\dots,x_k)$.
    Since $\sum_{i=k+1}^K|L_i|\le h$, the $(h,c)$-list recoverability ensures that at most $c$ choices of $(x_1,\dots,x_k)$ satisfy $C(x_1,\dots,x_k)\in L_{k+1}\times\dots\times L_K$.
    The remaining values $x_{k+1},\dots,x_K$ are then uniquely determined by $x_1,\dots,x_k$, proving the bound.
\end{proof}

\paragraph{Step 2: checking codewords by hashing.}
We use linear separators to hash $\F_q^\ell$ to a smaller space, so that we can use an \mmsalevel{4} formula to test whether a codeword is valid.
Fix an $\F_2$-basis of $\F_q$ to identify $\F_q^\ell$ with $\F_2^D$, where $D:=\ell\log q$.
By \cref{lem:linear-separator}, construct a $(D,d,h^{K})$-linear separator $\{\psi_r:\F_2^D\to\Sigma\}_{r\in[m]}$, where $\Sigma:=\F_2^d$ and
$$
    m=O(Dh^{2K}),\qquad p:=|\Sigma|=O(h^{K}).
$$
For every $i\in[K]$ and $\alpha\in\F_q$, introduce a Boolean variable $x_{i,\alpha}$ indicating whether $v_{i,\alpha}$ and its associated set are selected.
For each $r\in[m]$, $i\in[K]$, and $\sigma\in\Sigma$, let $y^r_{i,\sigma}:=\OR{\alpha\in\F_q:\,\psi_r(v_{i,\alpha})=\sigma}x_{i,\alpha}$.
For $z\in\{0,1\}^{[K]\times\Sigma}$, define
$$
    \Gamma^r(z):=\bigg\{b\in\Sigma^{K}:
        \sum_{i\in[K]} b_i=\vec0,\,\AND{i\in[K]}z_{i,b_i}=1\bigg\}.
$$
Thus $y^r$ records which hash values appear among the selected vectors in each group.
The set $\Gamma^r(y^r)$ lists all $K$-tuples of these hash values whose sum is $\vec 0$. This is a necessary condition for the $K$ selected vectors to sum to $\vec0$.
Each selected vector contributes to one entry of $y^r$, so $\|y^r\|_0\le\|x\|_0$.
Define
$$
    F_r(x,z):=\OR{b\in\Gamma^r(z)} \AND{e\in U_0}
        \OR{i\in[K],\,\alpha\in \F_q:\\ \psi_r(v_{i,\alpha})=b_i,\,e\in S_{i,\alpha}} x_{i,\alpha},
$$
Let $\Phi_r(x):=F_r(x,y^r)$ and $\Phi(x):=\AND{r\in[m]}\Phi_r(x)$.
We call each $\Phi_r$ a test $r\in[m]$.

\begin{claim}\label{clm:modular-hashing}
    The formula $\Phi$ has the following properties:
    \begin{itemize}
        \item \emph{(Completeness.)} If $\opt(I_0)\le k$, then $\opt(\Phi)\le K$.
        \item \emph{(Soundness.)} If $\opt(I_0)>k$, then every $x\in\{0,1\}^{[K]\times\F_q}$ with $\|x\|_0\le h$ has a test $r\in[m]$ satisfying
        $\Phi_r(x)=0$ and $|\Gamma^r(y^r)|\le c$.
    \end{itemize}
\end{claim}

\begin{proof}
    We first show the completeness. Suppose $\opt(I_0)\le k$. \Cref{clm:modular-encoding} gives $\alpha\in\F_q^K$ with $\bigcup_{i=1}^K S_{i,\alpha_i}=U_0$ and $\sum_{i=1}^K v_{i,\alpha_i}=\vec 0$. Set $x_{i,\alpha_i}=1$ for every $i\in[K]$ and all other variables to $0$, giving $\|x\|_0=K$.
    Fix a test $r$ and let $b:=(\psi_r(v_{1,\alpha_1}),\dots,\psi_r(v_{K,\alpha_K}))$. Each $y^r_{i,b_i}=1$, and linearity gives $\sum_{i=1}^K b_i=\psi_r\big(\sum_{i=1}^K v_{i,\alpha_i}\big)=\vec 0$, thus $b\in\Gamma^r(y^r)$.
    Since $\bigcup_{i=1}^K S_{i,\alpha_i}=U_0$, every $e\in U_0$ belongs to some $S_{i,\alpha_i}$ with $x_{i,\alpha_i}=1$ and $\psi_r(v_{i,\alpha_i})=b_i$.
    Hence $\Phi_r(x)=1$ for every $r$, so $\Phi(x)=1$ and $\opt(\Phi)\le K$.

    We next prove soundness. Suppose $\opt(I_0)>k$. Fix any $x$ with $\|x\|_0\le h$ and let $L_i:=\{\alpha\in\F_q:x_{i,\alpha}=1\}$.
    Consider
    $$
        Z:=\bigg\{\sum_{i=1}^K v_{i,\alpha_i}:\alpha_i\in L_i\bigg\} \setminus\{\vec0\}.
    $$
    Since $|Z|\le\prod_i|L_i|\le h^{K}$, some test $r\in[m]$ satisfies
    $\psi_r(z)\ne\vec0$ for every $z\in Z$.
    Fix such a test $r$.
    For every $\alpha_1\in L_1,\dots,\alpha_K\in L_K$, linearity and the choice of $r$ give
    $$
        \sum_{i=1}^K\psi_r(v_{i,\alpha_i})=\vec 0
        \quad\implies\quad
        \sum_{i=1}^K v_{i,\alpha_i}=\vec 0.
    $$

    We first show that $|\Gamma^r(y^r)|\le c$.
    For every $b\in\Gamma^r(y^r)$, choose any $\alpha_i\in L_i$ satisfying $\psi_r(v_{i,\alpha_i})=b_i$ for every $i\in[K]$.
    Since $\sum_i b_i=\vec 0$, we have $\sum_{i=1}^K v_{i,\alpha_i}=\vec 0$.
    Every $(\alpha_1,\dots,\alpha_K)$ determines only one $b$, hence
    $$
        |\Gamma^r(y^r)| \le \bigg|\bigg\{
            (\alpha_1,\dots,\alpha_K)\in L_1\times\dots\times L_K:
            \sum_{i=1}^K v_{i,\alpha_i}=\vec0
        \bigg\}\bigg| \le c.
    $$
    The last inequality follows from \cref{clm:modular-encoding}.

    To prove $\Phi_r(x)=0$, we show that, for every
    $b\in\Gamma^r(y^r)$ and $i\in[K]$, there is exactly one $\alpha_i\in L_i$ satisfying $\psi_r(v_{i,\alpha_i})=b_i$.
    Fix $b\in\Gamma^r(y^r)$ and choose any $\alpha_i\in L_i$ with $\psi_r(v_{i,\alpha_i})=b_i$ for every $i\in[K]$.
    We have $\sum_{i=1}^K v_{i,\alpha_i}=\vec 0$.
    If $\alpha_i'\in L_i$ also satisfies $\psi_r(v_{i,\alpha_i'})=b_i$, then $\psi_r(v_{i,\alpha_i'})+\sum_{j\neq i}\psi_r(v_{j,\alpha_j})=\sum_{j} b_j=\vec 0$,
    and the choice of $r$ again implies
    $v_{i,\alpha_i'}+\sum_{j\neq i}v_{j,\alpha_j}=\vec 0$, hence $v_{i,\alpha_i}=v_{i,\alpha_i'}$.
    The uniqueness of vectors within each group implies $\alpha_i'=\alpha_i$.
    Consequently,
    $$
        \bigcup_{i=1}^K
        \bigcup_{\substack{\alpha\in L_i:\\
            \psi_r(v_{i,\alpha})=b_i}}S_{i,\alpha}
        =\bigcup_{i=1}^K S_{i,\alpha_i}
        =\bigcup_{i=1}^k S_{\alpha_i}\neq U_0.
    $$
    Thus $\Phi_r(x)=0$, as required.
\end{proof}

\paragraph{Step 3: converting to \mmsalevel{2}.}
We convert $\Phi$ to an \mmsalevel{2} formula while preserving completeness and soundness.
For every $r\in[m]$, let
$$\begin{aligned}
    \Psi_r(x) &:=
    \AND{z\in\{0,1\}^{[K]\times\Sigma}:\\ \|z\|_0\le h,\,|\Gamma^r(z)|\le c}
    \bigg(
        F_r(x,z) \vee
        \OR{i\in[K],\,\sigma\in \Sigma:\\ z_{i,\sigma}=0} y^r_{i,\sigma}
    \bigg).
\end{aligned}$$
Let $\Psi(x):=\AND{r\in[m]}\Psi_r(x)$.

\begin{claim}\label{clm:modular-conversion}
    For every test $r$ and assignment $x$, $\Psi_r(x)\ge\Phi_r(x)$, with equality whenever $\|y^r\|_0\le h$ and $|\Gamma^r(y^r)|\le c$.
    Moreover, $\Psi$ is equivalent to an \mmsalevel{2} formula on the same variables.
\end{claim}

\begin{proof}
For fixed $x$, the function $F_r(x,z)$ is monotone in $z$.
    Applying \cref{lem:local-acceptance-general} with $\Lcal=\{z\in\{0,1\}^{[K]\times\Sigma}:\|z\|_0\le h,|\Gamma^r(z)|\le c\}$ gives $\Psi_r(x)\ge F_r(x,y^r)=\Phi_r(x)$ for every $x$, with equality when $\|y^r\|_0\le h$ and $|\Gamma^r(y^r)|\le c$.

    Expanding $F_r$ and applying the hypercube set system to $\Psi_r$ by swapping $\OR{b\in\Gamma^r(z)}$ and $\AND{e\in U_0}$ gives
    $$\begin{aligned}
        \Psi_r(x)
        &= \AND{z\in\{0,1\}^{[K]\times\Sigma}:\\ \|z\|_0\le h,\,|\Gamma^r(z)|\le c}\Bigg(\OR{b\in\Gamma^r(z)}\AND{e\in U_0}\OR{i\in[K],\,\alpha\in\F_q:\\ \psi_r(v_{i,\alpha})=b_i,\,e\in S_{i,\alpha}}x_{i,\alpha}\vee\OR{i\in[K],\,\sigma\in\Sigma:\\ z_{i,\sigma}=0}y^r_{i,\sigma}\Bigg)\\
        &= \AND{z\in\{0,1\}^{[K]\times\Sigma}:\\ \|z\|_0\le h,\,|\Gamma^r(z)|\le c}\AND{f:\Gamma^r(z)\to U_0}\Bigg(\OR{b\in\Gamma^r(z)}\OR{i\in[K],\,\alpha\in\F_q:\\ \psi_r(v_{i,\alpha})=b_i,\,f(b)\in S_{i,\alpha}}x_{i,\alpha}\vee\OR{i\in[K],\,\alpha\in\F_q:\\ z[i,\psi_r(v_{i,\alpha})]=0}x_{i,\alpha}\Bigg).
    \end{aligned}$$
    Therefore $\Psi=\AND{r\in[m]}\Psi_r$ is an \mmsalevel{2} formula.
\end{proof}

Finally converting $\Psi$ to $K$-\setcover by \cref{def:setcover} gives the following reduction.

\begin{lemma}\label{lem:setcover-gap-reduction-modular}
    For the inputs fixed above, the construction produces a $K$-\setcover instance $I'=(\Scal',U')$ with the following properties:
    \begin{itemize}
        \item \emph{(Size.)} $|\Scal'|= Kq$ and $|U'|\le h^{O(Kh)}u_0^c\log q$.
        \item \emph{(Completeness.)} If $\opt(I_0)\le k$, then $\opt(I')\le K$.
        \item \emph{(Soundness.)} If $\opt(I_0)>k$, then $\opt(I')>h$.
        \item \emph{(Running time.)} The construction takes $(qh^{Kh}u_0^{c})^{O(1)}$ time.
    \end{itemize}
\end{lemma}

\begin{proof}
    Completeness follows from \cref{clm:modular-hashing,clm:modular-conversion}.
    In the NO case, for every $x$ with $\|x\|_0\le h$, \cref{clm:modular-hashing} gives a test $r$ satisfying $\Phi_r(x)=0$ and $|\Gamma^r(y^r)|\le c$.
    Since $\|y^r\|_0\le\|x\|_0$, \cref{clm:modular-conversion} gives
    $\Psi_r(x)=\Phi_r(x)=0$ at this test, proving soundness.
    It remains to bound the size and running time.

    There are at most $(Kp)^h$ choices of $z$ with
    $\|z\|_0\le h$ per test.
    Each choice with $|\Gamma^r(z)|\le c$ contributes at most $u_0^{|\Gamma^r(z)|}\le u_0^c$ clauses. Since $p=O(h^K)$ and $m=O(Dh^{2K})=h^{O(K)}\log q$, we have
    $$
        |U'|\le m(Kp)^h u_0^c
        \le h^{O(Kh)}u_0^c\log q
    $$
    The formula has $Kq$ variables, so $|\Scal'|= Kq$.

    Constructing the separator takes $(Dh^K)^{O(1)}$ time.
    There are $|\Sigma|^K=h^{O(K^2)}$ choices of $b\in\Sigma^K$, which can be enumerated in $h^{O(Kh)}$ time since $h\ge K$.
    Enumerating all $b,z$, and the functions $f:\Gamma^r(z)\to U_0$ takes $(k\ell qu_0)^{O(1)}h^{O(Kh)}u_0^c=(qh^{Kh}u_0^{c})^{O(1)}$ time in total.
\end{proof}

We now choose the code and parameters using the construction in \cref{lem:additive-code}.

\begin{theorem}\label{thm:gap-amplification-modular}
    There is a constant $\alpha>0$ with the following property.
    Given a $k$-\setcover instance $I_0=(\Scal_0,U_0)$ with $n:=|\Scal_0|$ and $u_0:=|U_0|\le n$, let $\rho,\gamma$ satisfy
    $$
        u_0^{\Omega(1)}\le \rho\le n,\qquad
        1\le\gamma\le \alpha\min\bigg\{
            \frac{\log\rho}{\log u_0},
            \frac{\log\rho}{k^2\log\log\rho}
        \bigg\}.
    $$
    There is a randomized reduction from $I_0$ to a $\gamma$-gap $K$-\setcover instance $I'=(\Scal',U')$ with $K=\Theta(k)$, $|\Scal'|=\Theta(kn)$ and $|U'|\le \rho\log n$ in $n^{O(1)}$ time. Completeness holds always, and soundness holds with probability at least $1-n^{-1}$.
\end{theorem}

\begin{proof}
    Let $\lambda,\beta$ be the constants in \cref{lem:additive-code}, and $q$ be the smallest power of two satisfying $q\ge n$.
    Let
    $$
        \ell:=\lambda k,\quad
        K:=k+\ell,\quad
        h:=\gamma K,\quad
        c:=\Big\lfloor \frac{h}{k}\Big\rfloor.
    $$
    Since $h=O(\gamma k)$ and $\rho\le n$, choosing sufficiently small $\alpha$ ensures $h\le \frac{\beta\log q}{k\log\log q}$.
    Sample the code $C$ in \cref{lem:additive-code} and apply \cref{lem:setcover-gap-reduction-modular} with these parameters. The code is injective and $(h,c)$-list recoverable with probability at least $1-q^{-1}\ge1-n^{-1}$.
    On this event, the reduction gives a $\gamma$-gap $K$-\setcover instance $(\Scal',U')$ with $|\Scal'|=Kq=\Theta(kn)$ and $|U'|\le h^{O(kh)}u_0^{h/k}\log n$.
    Since $h=O(k\gamma)$ and $\log h=O(\log\log\rho)$, for sufficiently small $\alpha$, we have
    $$
        |U'|\le 2^{O(kh\log h+(h/k)\log u_0)}\log n
        \le 2^{O(\alpha\log\rho)}\log n
        \le \rho\log n.
    $$
    Finally, since $h^{Kh}u_0^c=\rho^{O(1)}$ and $q=O(n)$, both the code construction and the reduction take $n^{O(1)}$ time.
\end{proof}

\subsection{\ensuremath{k}-OV Hardness of \ensuremath{k}-\setcover}
The $k$-\ov problem in dimension $O(\log n)$ is equivalent to multicolored $k$-\setcover with universe size $O(\log n)$.
We first reduce $k$-\ov to small-universe $k$-\setcover, and then apply \cref{thm:gap-amplification-modular} to obtain the following fine-grained inapproximability result under the randomized $k$-\ov Hypothesis.

\begin{theorem}[\Cref{thm:intro-kov-hardness} restated]
\label{thm:kov-hardness}
    Assume the randomized $k$-\textsc{OV} Hypothesis.
    There are constants
    $\alpha,\varepsilon>0$ such that no bounded-error randomized algorithm solves
    $$
        \frac{\alpha\log N}{K^2\log\log N}
        \text{-gap $K$-\setcover}
    $$
    on instances $(\Scal,U)$ with $|\Scal|=N$ and $|U|\le N$
    in $T(K)N^{\varepsilon K}$ time, for any computable function $T$.
\end{theorem}
\begin{proof}
    Fix a sufficiently large constant $k$. By \cref{hyp:kov}, $k$-\ov with $m$ vectors per group and dimension $d=O(\log m)$ cannot be solved with bounded error in $O(m^{k/2})$ time.

    We first reduce a $k$-\ov instance to a small-universe $k$-\setcover instance.
    Start from an instance $X_1,\dots,X_k\subseteq\{0,1\}^d$ with $|X_i|=m$ and $d=O(\log m)$.
    Introduce $k$ elements $x_1,\dots,x_k$ denoting each group and let
    $$
        U_0:=[d]\union\{x_1,\ldots,x_k\},\qquad
        S_{i,a}:=\{x_i\}\cup\{j\in[d]:a[j]=0\}
        \quad(a\in X_i).
    $$
    Covering all $x_1,\dots,x_k$ with at most $k$ sets forces exactly one choice from each group. The remaining elements are covered exactly when the chosen vectors are orthogonal.
    Thus $(\Scal_0,U_0)$ where $\Scal_0:=\{S_{i,a}:i\in[k],a\in X_i\}$ has a $k$-cover if and only if the $k$-\ov instance is a YES instance. Moreover, we have $n:=|\Scal_0|=km$ and $u_0:=|U_0|=d+k=O(\log n)$.

    Let $\lambda$ be the constant in \cref{lem:additive-code}.
    Let $\ell:=\lambda k$, $K:=k+\ell$, and $q$ be the smallest power of two with $q\ge n$.
    Applying \cref{thm:gap-amplification-modular} with $\rho:=n/\log n$ gives a $\gamma$-gap $K$-\setcover instance $(\Scal,U)$ with $N:=|\Scal|=\Theta(kn)$, $|U|\le n\le N$ and
    $$
        \gamma=\Omega\bigg(\min\bigg\{
            \frac{\log\rho}{\log u_0},
            \frac{\log\rho}{k^2\log\log\rho}
        \bigg\}\bigg)
        =\Omega\bigg(\frac{\log N}{K^2\log\log N}\bigg).
    $$
    Choose $\alpha$ small enough to obtain the stated gap.

    The reduction takes $m^{C_0}$ time for some absolute constant $C_0$.
    A bounded-error randomized algorithm for $\gamma$-gap $K$-\setcover with running time $T(K)N^{\varepsilon K}$ would therefore solve $k$-\ov with bounded error in time
    $$
        O(m^{C_0})+T'(k)m^{\varepsilon (1+\lambda)k} \le O(m^{k/2})
    $$
    for $\varepsilon<\frac{1}{2(1+\lambda)}$ and a sufficiently large fixed $k$.
    This contradicts \cref{hyp:kov}.
\end{proof}

\subsection{From Large Gap to Small Universe}
We present a randomized Turing reduction in the reverse direction based on the element-sampling argument of
\cite[Lemma~1]{DIMV14}. The key idea of the algorithm is to repeatedly cover a sample of the uncovered elements using at most $k$ sets.
\begin{theorem}
\label{thm:reverse-reduction}
    Let $\Ocal$ be an oracle for exact $k$-\setcover, which on input a $k$-\setcover instance $(\Scal,T)$, outputs whether there exists $\Ccal\subseteq\Scal$ with $|\Ccal|\le k$ and $T\subseteq\bigcup_{S\in\Ccal}S$.
    There is a randomized algorithm that, on input a $k$-\setcover instance $I=(\Scal,U)$ with $n:=|\Scal|$ and $u:=|U|\le n$, and an integer $\gamma$ with $1\le\gamma\le O(\log n)$, satisfies the following properties:
    \begin{itemize}
        \item \emph{(Query size.)} Every query has at most $n$ sets and universe size at most $O(u^{1/\gamma}k\log n)$.
        \item \emph{(Completeness.)} The algorithm always outputs YES if $\opt(I)\le k$.
        \item \emph{(Soundness.)} The algorithm outputs NO with probability at least $2/3$ if $\opt(I)>\gamma k$.
        \item \emph{(Running time.)} The reduction makes $O(\gamma n)$ oracle calls and takes $n^{O(1)}$ additional time.
    \end{itemize}
\end{theorem}
\begin{proof}
    Let
    $$
        p:=u^{-1/\gamma},\qquad
        u_0:=\big\lceil p^{-1}\ln(3\gamma n^k)\big\rceil
        = O(u^{1/\gamma}k\log n).
    $$
    The algorithm initially sets $R_0:=U$, and then runs $\gamma$ rounds. In round $i\in[\gamma]$:
    \begin{itemize}
        \item If $R_{i-1}=\emptyset$, then return YES.
        \item Otherwise, choose a uniform random subset $T_i\subseteq R_{i-1}$ of size $|T_i|=\min\{u_0,|R_{i-1}|\}$. Then query $\Ocal$ on instance $(\Scal,T_i)$, and make decisions based on its output:
        \begin{itemize}
            \item Return NO if $\Ocal(\Scal,T_i)$ outputs NO.
            \item Otherwise, recover $\Ccal_i\subseteq\Scal$ with $|\Ccal_i|\le k$ and $T_i\subseteq\bigcup_{S\in\Ccal_i}S$ using at most $n$ additional queries: starting with $\Scal$, delete each set whenever the remaining family still admits a $k$-cover of $T_i$, and the final family is such a cover $\Ccal_i$. Let $R_i:=R_{i-1}\setminus\bigcup_{S\in\Ccal_i}S$ denote the remaining universe not covered by $\Ccal_i$.
        \end{itemize}
    \end{itemize}
    The algorithm outputs YES after $\gamma$ rounds.

    \subparagraph{Completeness.}
    If $\opt(I)\le k$, a cover of $U$ using at most $k$ sets covers every sample $T_i$. Thus every query $\Ocal(\Scal,T_i)$ returns YES, and the algorithm always outputs YES.

    \subparagraph{Soundness.}
    Assume $\opt(I)>\gamma k$.
    Fix any round $i\in[\gamma]$ and its remaining universe $R:=R_{i-1}\neq\varnothing$.
    For $\Ccal\subseteq\Scal$, let $R(\Ccal):=R\setminus\bigcup_{S\in\Ccal}S$ denote the remaining universe after selecting $\Ccal$.
    We show that, with probability at least $1-\frac{1}{3\gamma}$ over the random choice of $T_i$, every $k$-cover $\Ccal$ of $T_i$ satisfies $|R(\Ccal)|<p|R|$.

    Write $T:=T_i$. If $|R|\le u_0$, then $T=R$, thus every cover of $T$ covers $R$. Otherwise, $|T|=u_0$.
    Fix $\Ccal\subseteq\Scal$ with $|\Ccal|\le k$ and $|R(\Ccal)|\ge p|R|$.
    It covers $T$ if and only if $T\subseteq R\setminus R(\Ccal)$, which has probability
    $$
        \Pr_T\big[\text{$\Ccal$ covers $T$}\big]
        = \frac{\binom{|R|-|R(\Ccal)|}{u_0}}{\binom{|R|}{u_0}}
        \le \bigg(1-\frac{|R(\Ccal)|}{|R|}\bigg)^{u_0}
        \le (1-p)^{u_0}
        \le \e^{-p u_0}.
    $$
    There are at most $n^k$ nonempty families $\Ccal\subseteq\Scal$ with $|\Ccal|\le k$.
    A union bound over all choices of $\Ccal$ shows that the probability that any such family covers $T$ but leaves at least $p|R|$ elements uncovered is at most
    $$
        n^k\e^{-p u_0}\le\frac1{3\gamma}.
    $$
    Thus, regardless of the choice of $\Ccal_i$, with probability at least $1-\frac{1}{3\gamma}$, the round $i$ either rejects or satisfies $|R_i|<p|R_{i-1}|$.

    Let $E_i$ be the event that round $i$ returns a cover $\Ccal_i$ with $|R_i|\ge p|R_{i-1}|$. Since the preceding bound holds for every nonempty $R_{i-1}$, $\Pr[E_i]\le\frac{1}{3\gamma}$. Let $E:=\bigcup_{i\in[\gamma]}E_i$, then a union bound gives
    $$
        \Pr[E] \le \sum_{i\in[\gamma]}\Pr[E_i] \le \frac{1}{3}.
    $$
    We show by contradiction that if $E$ does not happen, then the algorithm outputs NO.
    If the algorithm outputs YES at the $i$-th round, the algorithm guarantees that $R_{i-1}=\emptyset$.
    If the algorithm outputs YES after $\gamma$ rounds, then $|R_\gamma|<p^\gamma u=1$, thus $R_\gamma=\emptyset$.
    In both cases, there exists $t\in[\gamma]$ such that $R_t=\emptyset$. Hence $\Ccal:=\bigcup_{i=1}^{t}\Ccal_i$ covers $U$, and $|\Ccal|\le tk\le\gamma k$, contradicting $\opt(I)>\gamma k$.
    Therefore, the algorithm outputs NO with probability at least $1-\Pr[E]\ge\frac{2}{3}$.

    \subparagraph{Query size and running time.}
    Every query has at most $n$ sets and at most $\min\{u,u_0\}$ elements, and there are at most $\gamma(n+1)$ queries.
    Since $\gamma=O(\log n)$ and $u\le n$, the additional running time is $n^{O(1)}$.
\end{proof}

Combining \cref{thm:gap-amplification-modular,thm:reverse-reduction} gives the following equivalence between small-universe $k$-\setcover and large-gap $k$-\setcover.

\begin{corollary}[\Cref{thm:intro-equivalence} restated]
\label{cor:randomized-equivalence}
    Let $n,u,k\in\N$ and $\gamma$ satisfy $u\le n$ and
    $$
        1\le \gamma\le \min\bigg\{
            \frac{\log u}{\log\log n},
            \frac{\log u}{k^2\log\log u}
        \bigg\}.
    $$
    The following reductions hold:
    \begin{itemize}
        \item Given an exact $k$-\setcover instance with $n$ sets and universe size at most $u^{1/\gamma}$,
        there is a randomized reduction to an $\Omega(\gamma)$-gap $O(k)$-\setcover instance with $O(kn)$ sets and universe size at most $u$, in $n^{O(1)}$ time.

        \item Given a $\gamma$-gap $k$-\setcover instance with $n$ sets and universe size at most $u$, there is a randomized Turing reduction to exact $k$-\setcover instances, each with at most $n$ sets and universe size $u^{O(1/\gamma)}$. The reduction makes $O(\gamma n)$ oracle calls and uses $n^{O(1)}$ additional time.
    \end{itemize}
    Thus exact $k$-\setcover with universe size $u^{\Theta(1/\gamma)}$ and $\gamma$-gap $k$-\setcover with universe size $u$ are equivalent under randomized $\FPT$ Turing reductions.
\end{corollary}

\begin{proof}
    For the forward direction, let $I_0=(\Scal_0,U_0)$ be the exact $k$-\setcover instance with $u_0:=|U_0|\le u^{1/\gamma}$.
    We assume that $\gamma$ is sufficiently large, since $I_0$ itself gives the required $\Omega(\gamma)$ gap when $\gamma$ is a constant.

    Let $\rho:=\frac{u}{\log n}$.
    By $\gamma\le \frac{\log u}{\log\log n}$, we have $\log n\le u^{1/\gamma}$.
    For $\gamma\ge 2$, we have $\rho\ge u^{1-1/\gamma}\ge u^{1/\gamma}\ge u_0$ and $\rho\le n$.
    Hence \cref{thm:gap-amplification-modular} gives a randomized reduction from $I_0$ to an $O(k)$-\setcover instance $I=(\Scal,U)$ with $|\Scal|=\Theta(kn)$ and $|U|\le u$ in $n^{O(1)}$ time, whose approximation gap is at least
    $$\begin{aligned}
        \alpha\min\bigg\{
            \frac{\log\rho}{\log u_0},
            \frac{\log\rho}{k^2\log\log\rho}
        \bigg\}
        &\ge \alpha\min\bigg\{
            \frac{(1-1/\gamma)\log u}{(1/\gamma)\log u},
            \frac{(1-1/\gamma)\log u}{k^2\log\log u}
        \bigg\} \\
        &\ge \alpha(\gamma-1)
        =\Omega(\gamma).
    \end{aligned}$$

    For the reverse direction, let $I=(\Scal,U)$ be the $\gamma$-gap $k$-\setcover instance and let $r:=\lfloor\gamma\rfloor$.
    Since $r\le O(\log n)$, applying \cref{thm:reverse-reduction} gives a randomized Turing reduction to exact $k$-\setcover that makes $O(\gamma n)$ oracle calls and uses $n^{O(1)}$ additional time.
    Each query has at most $n$ sets and universe size at most $O(u^{1/r}k\log n)$.
    The upper bound for $\gamma$ gives $\log n\le u^{1/\gamma}$ and $k \le\log n$.
    Since $r\ge\gamma/2$, every query therefore has universe size
    at most $O(u^{1/r}k\log n)=u^{O(1/\gamma)}$.
\end{proof}

    \section{Inapproximability of \ensuremath{k}-\setcover with Small VC Dimension}
\label{sec:vc-dimensions}

\subsection{VC Dimensions of the Reduction in \texorpdfstring{\cref{sec:klm}}{Section 3}}
We analyze the VC dimensions of the reductions of \cref{thm:vectorsum-to-labelcover,thm:labelcover-to-setcover}.
Given $h\ge k\ge 2$, a $k$-\vectorsum instance $(X_1,\dots,X_k)$ where $X_i=\{v_{i,a}:a\in[n]\}\subseteq\F_2^D$ for $i\in[k]$, and a $(D,\ell,A)$-linear separator $\Hcal=\{\psi_r\}_{r\in[m]}$ with $\psi_r:\F_2^D\to\Sigma:=\F_2^\ell$,
the neighborhoods in \cref{thm:vectorsum-to-labelcover} are
$$
    P_{r,i,a}:=\{b\in\Sigma^k:\textstyle\sum_{j=1}^k b_j=\vec0,\,b_i=\psi_r(v_{i,a})\}.
$$
The output \mmsalevel{2} instance of \cref{thm:labelcover-to-setcover} is
$$
    \Psi(x):=\AND{r\in[m]}\AND{z\in\{0,1\}^{\Pcal^r}:\\\|z\|_0\le h,\,F_r(z)=0}
        \ \OR{i\in[k],\,a\in[n]:\\z[i,P_{r,i,a}]=0}x_{i,a},
$$
where $\Pcal^r:=\{(i,P_{r,i,a}):i\in[k],\ a\in[n]\}$ and $F_r:\{0,1\}^{\Pcal^r}\to\{0,1\}$ is the function $F$ for test $r$ defined in the proof of \cref{thm:labelcover-to-setcover}.
Its corresponding $k$-\setcover instance is
$$\begin{aligned}
    U_r&:=\{z\in\{0,1\}^{\Pcal^r}:\|z\|_0\le h,\,F_r(z)=0\},\quad
    &U&:=\{(r,z):r\in[m],\,z\in U_r\},\\
    S_{i,a}&:=\{(r,z)\in U:z[i,P_{r,i,a}]=0\},\quad
    &\Scal&:=\{S_{i,a}:i\in[k],\,a\in[n]\}.
\end{aligned}$$
Write $d$ and $d^*$ for the primal and dual VC dimensions of $(\Scal,U)$, respectively.
We define a set system over universe $U_T:=\{(r,i,P):r\in[m],\,(i,P)\in\Pcal^r\}$:
$$
    T_{i,a}:=\{(r,i,P_{r,i,a}):r\in[m]\},\qquad
    \Tcal:=\{T_{i,a}:i\in[k],\,a\in[n]\}.
$$
Write $d_T$ for the primal VC dimension of $\Tcal$.
To bound $d$, we need the following separator which bounds $d_T$.

\begin{lemma}\label{lem:vc-separator-graphs}
    For $D,\ell\in\N$~and $A\ge1$ with $2^\ell\ge DA$, there is a $(D,\ell,A)$-linear separator $\Hcal$ with $|\Hcal|\le\lceil DA/\ell\rceil$, constructible deterministically in $2^{O(\ell)}$ time.
    Moreover, for $P_{r,i,a}$ and $T_{i,a}$ defined above using this separator, $d_T\le D/\ell+1$.
\end{lemma}
\begin{proof}
    We construct the $(D,\ell,A)$-linear separator $\Hcal$ by \cref{lem:linear-separator-1} with $|\Hcal|\le \lceil DA/\ell\rceil=O(DA)$ in $(DA2^\ell)^{O(1)}=2^{O(\ell)}$ time, and for every $z\in\F_2^D\setminus\{\vec 0\}$, $|\{r\in[m]:\psi_r(z)=\vec0\}|<D/\ell$.

    To bound $d_T$, note that the definition of $P_{r,i,a}$ gives $P_{r,i,a}=P_{r,i,a'}$ if and only if $\psi_r(v_{i,a})=\psi_r(v_{i,a'})$. Thus, for $a\ne a'$, the definition of $T_{i,a}$ and linearity give
    $$\begin{aligned}
        |T_{i,a}\cap T_{i,a'}|
        &=|\{r\in[m]:P_{r,i,a}=P_{r,i,a'}\}|\\
        &=|\{r\in[m]:\psi_r(v_{i,a}-v_{i,a'})=\vec0\}|<D/\ell.
    \end{aligned}$$
    Since $T_{i,a}\cap T_{i',a'}=\varnothing$ when $i\neq i'$, any two distinct members of $\Tcal$ intersect in at most $D/\ell$ elements. By \cref{lem:vc-intersections}, $d_T\le D/\ell+1$.
\end{proof}

\begin{lemma}\label{lem:vc-dimensions-acceptance}
    For arbitrary neighborhoods $P_{r,i,a}$, functions $F_r:\{0,1\}^{\Pcal^r}\to\{0,1\}$, and $h\ge1$, let $d_T$ be the VC dimension of $\Tcal$.
    The primal and dual VC dimensions of $(\Scal,U)$ satisfy
    $$
        d=O\bigl(d_T\log h\bigr),\qquad
        d^*=O(h+\log m).
    $$
\end{lemma}
\begin{proof}
    To bound $d$, let $X:=\{(r_j,z_j):j\in[s]\}\subseteq U$ be a set of size $s$ shattered by $\Scal$. For each $j\in[s]$, define
    $$
        Y_j:=\{(r_j,i,P_{r_j,i,a}): i\in[k],\,a\in[n],\,z_j[i,P_{r_j,i,a}]=1\}\subseteq U_T,
    $$
    and let $Y:=\bigcup_{j\in[s]}Y_j$.
    Since $\|z_j\|_0\le h$, we have $|Y_j|\le h$ and $|Y|\le hs$.
    By definition, for every $i\in[k]$, $a\in[n]$, we have
    $$
        T_{i,a}\cap Y
        =\{(r_j,i,P_{r_j,i,a}):j\in[s],\,z_j[i,P_{r_j,i,a}]=1\}
    $$
    and
    $$
        S_{i,a}\cap X
        = \{(r_j,z_j):j\in[s],\,z_j[i,P_{r_j,i,a}]=0\}.
    $$
    If $S_{i,a}\cap X\ne S_{i',a'}\cap X$, then some $j\in[s]$ satisfies $z_j[i,P_{r_j,i,a}]\ne z_j[i',P_{r_j,i',a'}]$. Hence $(i,P_{r_j,i,a})\neq(i',P_{r_j,i',a'})$, and thus $T_{i,a}\cap Y\ne T_{i',a'}\cap Y$.
    Since $X$ is shattered, we have
    $$
        2^s=|\{S\cap X:S\in\Scal\}|
        \le|\{T\cap Y:T\in\Tcal\}|\le\pi_{\Tcal}(hs).
    $$
    Suppose $s>4d_T$, then \cref{thm:vc-sauer} gives $2^s\le\pi_{\Tcal}(hs)\le\big(\frac{\e hs}{d_T}\big)^{d_T}$.
    Taking logarithms and using $\log_2\frac{s}{d_T}\le \frac{s}{2d_T}$ gives $s\le2d_T\log_2(\e h)$.
    Thus, either $s\le 4d_T$ or $s\le 2d_T\log_2(\e h)$, proving $d=O(d_T\log h)$.

    To bound $d^*$, for each $r\in[m]$, define
    $$
        \Scal^*_r:=\{\{S\in\Scal:(r,z)\in S\}:z\in U_r\}.
    $$
    Since $\Scal^*=\bigcup_{r\in[m]}\Scal^*_r$, it suffices by \cref{lem:vc-union} to show that each $\Scal^*_r$ has VC dimension at most $h$.

    Fix $r\in[m]$, and let $\Xcal:=\{S_{i_1,a_1},\ldots,S_{i_s,a_s}\}\subseteq\Scal$ be a subfamily of size $s$ shattered by $\Scal^*_r$.
    By the definition of $S_{i_j,a_j}$, $(r,z)\in S_{i_j,a_j}$ if and only if $z[i_j,P_{r,i_j,a_j}]=0$ for any $z\in U_r$.
    For any distinct $j,j'\in[s]$, since $\Scal^*_r$ shatters $\Xcal$, there exists $z\in U_r$ such that $(r,z)\in S_{i_j,a_j}$ and $(r,z)\not\in S_{i_{j'},a_{j'}}$, which implies $z[i_j,P_{r,i_j,a_j}]=0$ and $z[i_{j'},P_{r,i_{j'},a_{j'}}]=1$.
    Hence the $s$ coordinates $(i_j,P_{r,i_j,a_j})$ are distinct.

    Since $\Scal^*_r$ shatters $\Xcal$, there is also some $z\in U_r$ such that $(r,z)\notin S_{i_j,a_j}$ for every $j\in[s]$.
    Thus $z[i_j,P_{r,i_j,a_j}]=1$ at all $s$ distinct coordinates, and therefore $s\le\|z\|_0\le h$.
\end{proof}

\subsection{\Wone-Hardness of \ensuremath{k}-\setcover with Bounded VC Dimensions}

\begin{lemma}\label{lem:vectorsum-vc}
    There exists a constant $\alpha>0$ with the following property.
    Given a $k$-\vectorsum instance $I=(X_1,\ldots,X_k)$ with $X_i\subseteq\F_2^D$, $|X_i|=n$, and $D=O(k\log n)$, let $h$ satisfy $k\le h\le\frac{\alpha\log n}{k\log\log n}$.
    There is a deterministic algorithm that constructs a $(k,h)$-gap \setcover instance $(\Scal,U)$ with $|\Scal|=kn$ and $|U|=n^{O(1)}$ in $n^{O(1)}$ time.
    Its primal and dual VC dimensions satisfy $d=O(kh\log h)$ and $d^*=O(h)$.
\end{lemma}
\begin{proof}
    Let
    $$
        A:=\bigg(\frac{h}{k}\bigg)^k,\qquad
        \ell:=\bigg\lceil\frac{\log_2 n}{h}\bigg\rceil,
        \qquad q:=2^\ell.
    $$
    By $h\le\frac{\alpha\log n}{k\log\log n}$, we have $q\ge(\log n)^{\Omega(k/\alpha)}$ and $DA\le (\log n)^{O(k)}$, hence $q\ge DA$ for sufficiently small $\alpha$.
    Construct a $(D,\ell,A)$-linear separator $\Hcal$ by \cref{lem:vc-separator-graphs}, with $|\Hcal|\le q$ and $d_T\le1+D/\ell$.
    Applying \cref{thm:vectorsum-to-labelcover} to $I$ with separator $\Hcal$, we obtain a \labelcover instance with parameters $m\le q$, $t=q^{k-1}$ and $p\le q$.
    Then apply \cref{thm:labelcover-to-setcover}, and write the resulting \mmsalevel{2} formula as a $k$-\setcover instance $(\Scal,U)$.
    The two reductions give a $(k,h)$-gap instance with $|\Scal|=kn$ and
    $$
        |U|\le m(kp)^h\le q(kq)^h=2^{O(h\log k+h\ell)}=n^{O(1)}.
    $$
    Constructing the separator and performing the two reductions take $(nDq^k)^{O(1)}(kq)^h=n^{O(1)}$ time in total.
    Finally, we have
    $$
        d_T\le 1+\frac{D}{\ell}\le 1+\frac{Dh}{\log_2 n}=O(kh),
        \qquad m\le 1+\frac{DA}{\ell}=O(kh)\cdot\bigg(\frac{h}{k}\bigg)^{k}\le 2^{O(h)}.
    $$
    Hence \cref{lem:vc-dimensions-acceptance} gives $d=O(d_T\log h)=O(kh\log h)$ and $d^*=O(h+\log m)=O(h)$.
\end{proof}

\begin{theorem}[\Cref{thm:intro-vc-hardness} restated]
\label{thm:vc-hardness}
    There exists a constant $\alpha>0$ such that, for every \FPT-computable $\gamma=\gamma(k,N)$ satisfying
    $$
        1\le\gamma\le\frac{\alpha\log N}{k^2\log\log N},
    $$
    $\gamma$-gap $k$-\setcover is $\Wone$-hard on instances $(\Scal,U)$ with $|\Scal|=N$, $|U|=N^{O(1)}$,
    and primal and dual VC dimensions satisfying $d=O(k^2\gamma\log(k\gamma))$ and $d^*=O(k\gamma)$.
\end{theorem}
\begin{proof}
    We start from a $k$-\vectorsum instance in \cref{lem:vectorsum-w1} with $|X_i|=n$ and $D=O(k\log n)$. Let $N:=kn$ and $h:=k\gamma(k,N)$.
    For sufficiently small $\alpha$, \cref{lem:vectorsum-vc} gives a $\gamma$-gap $k$-\setcover instance with $|\Scal|=N$ and $|U|=N^{O(1)}$ in $\FPT$ time.
    Moreover, $h=k\gamma$ gives $d=O(k^2\gamma\log(k\gamma))$ and $d^*=O(k\gamma)$.
\end{proof}

    \paragraph{AI methodology.}

Generative AI contributed to both the mathematical development and the preparation of this paper. In response to research questions posed by the authors, OpenAI’s GPT-5.6 Sol generated an initial proof of near-logarithmic inapproximability using repeated consistency tests and the enumeration of small rejecting collections of messages. This enumeration mechanism underlies the local acceptance family construction used in the proof of Theorem 1.1. OpenAI’s GPT-6 Astra generated an initial proof of the gap-amplifying reduction underlying the forward direction of Theorem 1.3, using additively separable encoding, compression of labels, and the enumeration of uncovered-element witnesses.
The authors verified the correctness and originality of all content, including references.

    \phantomsection
    \addcontentsline{toc}{section}{References}
    \bibliographystyle{alpha}
    \bibliography{ref}

    \appendix

    \section{Construction of Linear Separators}
\label{app:separators}

In this section, we prove \cref{lem:linear-separator}.
We first show two independent constructions of linear separators.\space
The separator condition depends only on $\lfloor A\rfloor$, so it suffices to prove the constructions for integer $A$.

\begin{lemma}\label{lem:linear-separator-1}
    For $D,d\in\N$~and $A\ge1$ with $2^d\ge DA$, there is a $(D,d,A)$-linear separator $\Hcal$ with $|\Hcal|\le\lceil DA/d\rceil$, constructible deterministically in $(DA2^d)^{O(1)}$ time. Moreover, for every $z\in\F_2^D\setminus\{\vec0\}$,
    $|\{\psi\in\Hcal:\psi(z)=\vec0\}|<D/d$.
\end{lemma}
\begin{proof}
    Fix an $\F_2$-linear isomorphism $\iota:\F_{2^d}\to\F_2^d$.
    For every $z\in\F_2^D$, pad $z$ with zeros and split it into $B:=\lceil D/d\rceil$ blocks of $d$ coordinates.
    Using $\iota^{-1}$, view these blocks as $\widehat z_0,\dots,\widehat z_{B-1} \in\F_{2^d}$, and define a polynomial $p_z$ over $\F_{2^d}$ by
    $$
        p_z(t):=\sum_{i=0}^{B-1}\widehat z_it^i
        \quad(t\in\F_{2^d}).
    $$
    Let $M:=\lceil DA/d\rceil$.
    Since $M\le DA\le2^d$, choose distinct $t_1,\dots,t_M\in\F_{2^d}$.
    For every $j\in[M]$, define the $\F_2$-linear map $\psi_j:\F_2^D\to\F_2^d$ by $\psi_j(z):=\iota(p_z(t_j))$.
    Let $\Hcal:=\{\psi_j:j\in[M]\}$, so $|\Hcal|\le M$.
    Constructing $\F_{2^d}$ and the $M$ matrices takes $(DA2^d)^{O(1)}$ time.

    To verify the separator property, fix any $Z\subseteq\F_2^D\setminus\{\vec0\}$ with $|Z|\le A$.
    For every $z\in Z$, the polynomial $p_z$ is nonzero and has degree at most $B-1<D/d$.
    Each $p_z$ has at most $B-1$ roots, hence
    $$
        |\{j\in[M]:\psi_j(z)=\vec0\}|
        =|\{j\in[M]:p_z(t_j)=0\}|
        \le B-1<D/d.
    $$
    Moreover,
    $$
        \bigg|\bigcup_{z\in Z}
        \{j\in[M]:p_z(t_j)=0\}\bigg|
        \le A(B-1)
        <\frac{DA}{d}\le M.
    $$
    Thus some $j\in[M]$ satisfies $\psi_j(z)\ne\vec0$ for every $z\in Z$, proving the separator property.
\end{proof}

\begin{lemma} \label{lem:linear-separator-2}
    For $D\in\N$~and $A\ge1$, there is a $(D,d,A)$-linear separator $\Hcal$ with $2^d=O(A)$ and $|\Hcal|=O(DA)$,
    constructible deterministically in $(2^DA)^{O(1)}$ time.
\end{lemma}
\begin{proof}
    Let $d:=\lceil\log_2 4A\rceil$ and $M:=2DA$.
    Fix an $\F_2$-linear isomorphism $\iota:\F_{2^d}\to\F_2^d$.
    We choose coefficients $a_{j,i}\in\F_{2^d}$ for $j\in[M]$ and $i\in[D]$ so that the maps
    $$
        \psi_j(z):=\iota\bigg(\sum_{i=1}^Da_{j,i}z_i\bigg)
        \quad(z\in\F_2^D)
    $$
    form a $(D,d,A)$-linear separator $\Hcal:=\{\psi_j:j\in[M]\}$.

    We assign the coefficients $a_{j,i}$ one at a time.
    At each step, we consider the probability space where all unassigned $a_{j,i}$ are chosen independently and uniformly from $\F_{2^d}$.
    Let $p_{j,z}:=\Pr[\psi_j(z)=\vec0]$, and consider the potential
    $$
        \Phi:=\sum_{z\in\F_2^D\setminus\{\vec0\}}
        2^{|\{j\in[M]:\psi_j(z)=\vec0\}|}.
    $$
    If some $a_{j,i}$ with $z_i=1$ is unassigned, then $p_{j,z}=2^{-d}$; otherwise, $p_{j,z}\in\{0,1\}$.
    For each fixed $z$, the values $\psi_j(z)$ are independent across $j$, so
    $$
        \E[\Phi]
        =\sum_{z\neq\vec0}\E\Big[2^{\sum_{j=1}^M[\psi_j(z)=\vec 0]}\Big]
        =\sum_{z\neq\vec0}\prod_{j=1}^M(1+p_{j,z}).
    $$
    For the next coefficient $a_{j,i}$, we try all values in $\F_{2^d}$ and choose one that does not increase $\E[\Phi]$. Such a value exists by averaging.
    Evaluating $\E[\Phi]$ takes $(2^DM)^{O(1)}$ time using rational numbers with $O(D+Md)$ bits.
    Since there are $MD$ coefficients and $2^d=O(A)$ possible values for each, all coefficients can be assigned in $(2^DA)^{O(1)}$ time.

    Before any coefficient is assigned, $p_{j,z}=2^{-d}$ for every $j,z$, so
    $$
        \E[\Phi]=(2^D-1)(1+2^{-d})^M
        <2^D\exp(M2^{-d})
        <2^{2D}=2^{M/A}.
    $$
    Thus after all coefficients are assigned, $\Phi<2^{M/A}$.

    To verify the separator property, fix any $Z\subseteq\F_2^D\setminus\{\vec0\}$ with $|Z|\le A$.
    If every $\psi_j$ maps some $z\in Z$ to zero, then some $z\in Z$ is mapped to zero by at least $M/A$ maps, giving $\Phi\ge2^{M/A}$, a contradiction.
    Hence some $\psi_j$ maps every $z\in Z$ to a nonzero vector, proving the separator property.
\end{proof}

We compose the two constructions to prove \cref{lem:linear-separator}.

\begin{lemma} \label{lem:linear-separator-concat}
    Given a $(D,d_1,A)$-linear separator $\Hcal_1$ and a $(d_1,d_2,A)$-linear separator $\Hcal_2$, one can construct a $(D,d_2,A)$-linear separator $\Hcal$ with $|\Hcal|\leq|\Hcal_1||\Hcal_2|$ in time polynomial in the input and output size.
\end{lemma}
\begin{proof}
    Let $\Hcal:=\{\psi_2\circ\psi_1:\psi_1\in\Hcal_1,\psi_2\in\Hcal_2\}$.
    For any $Z\subseteq\F_2^D\setminus\{\vec 0\}$ with $|Z|\le A$, there exists $\psi_1\in\Hcal_1$ such that $\vec 0\notin\psi_1(Z)$. Since $|\psi_1(Z)|\le |Z|\le A$, there exists $\psi_2\in\Hcal_2$ such that $\vec 0\notin\psi_2(\psi_1(Z))$. Hence $\psi_2\circ\psi_1\in\Hcal$ satisfies $\vec 0\notin(\psi_2\circ\psi_1)(Z)$.
\end{proof}

\begin{lemma}[\Cref{lem:linear-separator} restated]\label{lem:linear-separator-restated}
    For $D\in\N$~and $A\ge1$, there is a $(D,d,A)$-linear separator
    $\Hcal$ with $2^d=O(A)$ and $|\Hcal|=O(DA^2)$.
    Its matrices can be constructed deterministically in $(DA)^{O(1)}$ time.
\end{lemma}
\begin{proof}
    Choose the minimum positive integer $d_1$ with $2^{d_1}\ge DA$.
    By \cref{lem:linear-separator-1}, we construct a $(D,d_1,A)$-linear separator $\Hcal_1$ with $2^{d_1}=O(DA)$ and $|\Hcal_1|=O(DA/d_1)$ in $(DA)^{O(1)}$ time.
    Then by \cref{lem:linear-separator-2}, we construct a $(d_1,d,A)$-linear separator $\Hcal_2$ with $2^{d}=O(A)$ and $|\Hcal_2|=O(d_1A)$ in $(2^{d_1}A)^{O(1)}=(DA)^{O(1)}$ time.
    Composing $\Hcal_1$ and $\Hcal_2$ using \cref{lem:linear-separator-concat} gives a $(D,d,A)$-linear separator $\Hcal$ with $2^d=O(A)$ and $|\Hcal|=O(DA/d_1\cdot d_1A)=O(DA^2)$ in $(DA)^{O(1)}$ time.
\end{proof}

    \section{Construction of List-Recoverable Codes}
\label{app:encoding}

We use the following lemma given in \cite[Theorem~1(b)]{Thorup13} to show that the code $C$ constructed by randomly sampling $f_i$ is injective and list recoverable with high probability.

\begin{lemma}[{\cite{Thorup13}}]
\label{lem:tabulation-expansion}
    Let $q$ be a power of two and $k,\ell\in\N$ satisfy
    $\ell\ge6k$ and $(k+\ell)^k=q^{o(1)}$.
    For every $i\in[k]$ and $a\in\F_q$, choose $f_i(a)\in\F_q^\ell$ independently and uniformly, and define
    $$
        C(x_1,\ldots,x_k):=\sum_{i=1}^k f_i(x_i).
    $$
    For every fixed $\varepsilon\in(0,1)$, with probability
    $1-o(q^{2-\varepsilon\ell/(2k)})$, every nonempty
    $X\subseteq\F_q^k$ with $|X|\le q^{1/(5k)}$ satisfies
    $$
        \sum_{j=1}^{\ell}|\{C(x)_j:x\in X\}|>(1-\varepsilon)\ell|X|.
    $$
\end{lemma}

\begin{proof}[Proof of \cref{lem:additive-code}]
    Set $\lambda:=16$, $\ell:=\lambda k$, and $t:=\lfloor h/k\rfloor+1$.
    Sample every $f_i$ as in \cref{lem:tabulation-expansion}.
    Sampling these functions takes $(kq)^{O(1)}$ time.

    For sufficiently small $\beta$, the range of $h$ gives $(k+\ell)^k=(17k)^k=q^{o(1)}$ and $t\le q^{1/(5k)}$.
    Applying \cref{lem:tabulation-expansion} with $\varepsilon=1/2$, we obtain, with probability $1-o(q^{-2})\ge1-q^{-1}$, the following inequality holds simultaneously for every nonempty $X\subseteq\F_q^k$ with $|X|\le t$:
    $$
        \sum_{j=1}^{\ell}|\{C(x)_j:x\in X\}|>8k|X|.
    $$
    If $x_1\neq x_2\in\F_q^k$ satisfy $C(x_1)=C(x_2)$, taking $X=\{x_1,x_2\}$ makes the left-hand side equal to $\ell=16k$. Thus we have $16k>8k|X|=16k$, a contradiction. Hence $C$ is injective.

    To prove list recoverability, let $Y_1,\ldots,Y_\ell\subseteq\F_q$ satisfy $\sum_{j=1}^{\ell}|Y_j|\le h$.
    If $t$ distinct messages $x_1,\dots,x_t\in\F_q^k$ satisfy $C(x_i)\in Y_1\times\cdots\times Y_\ell$ for every $i\in[t]$, taking $X=\{x_1,\dots,x_t\}$ gives
    $$
        h\ge\sum_{j=1}^{\ell}|Y_j|
        \ge\sum_{j=1}^{\ell}|\{C(x)_j:x\in X\}|
        >8kt>h,
    $$
    a contradiction. Thus at most $t-1=\lfloor h/k\rfloor$ messages have their codewords in these lists.
\end{proof}

\end{document}